\documentclass[12pt]{article}
\usepackage{amssymb,amsmath,amsthm,amsfonts,eurosym,geometry,ulem,graphicx,caption,color,setspace,sectsty,comment,footmisc,caption,natbib,pdflscape,array,upgreek,bbm,verbatim,textpos,float,multirow,booktabs,caption,subcaption,tabularx,soul,siunitx}
\usepackage[bookmarks=true, bookmarksnumbered=true, hidelinks]{hyperref}           
\newcolumntype{Y}{>{\centering\arraybackslash}X}

\usepackage[para,online,flushleft]{threeparttable}
\usepackage{palatino}

\def\bs {\mathbf}

\def\bs{\boldsymbol}

\def\t{^{\top}}

\def\bs {\mathbf}

\def\bs{\boldsymbol}

\def\t{^{\top}}

\newtheorem{theorem}{Theorem}

\newtheorem{proposition}{Proposition}

\newtheorem{example}{Example}
\newtheorem{definition}{Definition}

\usepackage{natbib} 

\usepackage{etoolbox}

\newcounter{bibcount}
\makeatletter
\patchcmd{\@lbibitem}{\item[}{\item[\hfil\stepcounter{bibcount}{[\thebibcount]}}{}{}
\renewcommand\NAT@bibsetup%
[1]{\setlength{\leftmargin}{\bibhang}\setlength{\itemindent}{-\parindent}%
  \setlength{\itemsep}{\bibsep}\setlength{\parsep}{\z@}}
\makeatother

\newcommand{\R}{\mathbb{R}}

\newtheorem{assumption}{Assumption}

\newcolumntype{L}[1]{>{\raggedright\let\newline\\arraybackslash\hspace{0pt}}m{#1}}
\newcolumntype{C}[1]{>{\centering\let\newline\\arraybackslash\hspace{0pt}}m{#1}}
\newcolumntype{R}[1]{>{\raggedleft\let\newline\\arraybackslash\hspace{0pt}}m{#1}}

\begin{document}

\newgeometry{top=1.4in, bottom=1.4in, left=1.25in, right=1.25in}
\begin{titlepage}
  \title{\vspace*{-3cm}
\textbf{Conditional Choice Probability Estimation of Dynamic Discrete Choice Models with 2-period Finite Dependence}\thanks{We are very grateful for the comments from Vadim Marmer, Kevin Song and Florian Hoffmann.}
 }
  \author{
  	Yu Hao \footnote{\href{mailto: haoyu@hku.hk} {haoyu@hku.hk} Faculty of Business and Economics, The University of Hong Kong.} \\ \emph{The University of Hong Kong}
  \and
  Hiroyuki Kasahara \footnote{ \href{mailto: hkasahar@mail.ubc.ca}{hkasahar@mail.ubc.ca}Vancouver School of Economics, The University of British Columbia} 
  \\ \emph{The University of British Columbia}
} \date{ \today}\maketitle
  \begin{abstract}

This paper extends the work of \cite{arcidiacono2011conditional,arcidiacono_miller_2019nonstationary} by introducing a novel characterization of finite dependence within dynamic discrete choice models, demonstrating that numerous models display 2-period finite dependence.  We recast finite dependence as a problem of searching for weights sequentially, and introduce an computationally efficient method for determining these weights by utilizing the Kronecker product structure embedded in state transitions.   With the estimated weight, we develop a computationally attractive Conditional Choice Probability estimator with 2-period finite dependence. The computational efficacy of our proposed estimator is demonstrated through Monte Carlo simulations. 



\end{abstract}

  \setcounter{page}{0}
  \thispagestyle{empty}
\end{titlepage}


\section{Introduction}
To model the forward-looking nature of decision-making, economists often use \textit{ dynamic discrete choice models} (DDC) in many areas of economic research.\footnote{
	For example, in industrial organizations, \cite{Berry2006, Aguirregabiria2012, Sweeting2013, CollardWexler2013, Ellickson2015, Yakovlev2016}; in health economics, \cite{Gowrisankaran1997, Gowrisankaran2011, Gaynor2012, Beauchamp2015}; in the marketing,\cite{Dube2005, Doganoglu2006, Doraszelski2007}; and in labor economics, researchers use dynamic discrete choice to model the forward-looking nature of employment choice, school choice, and mammography choices \cite{Todd2006, Keane2011, Fang2015}.}
The paper by \cite{Rust1987, Rust1994} pioneered the estimation of dynamic discrete choice (DDC) models by identifying optimal decision rules and modeling future events' expectations. However, standard maximum likelihood estimation (MLE) for infinite-horizon discrete games can be computationally expensive due to the need for solving fixed-point problems in large state spaces. To address this issue, \cite{Hotz1993} proposed a computationally efficient method that represents the value function as a product of inverse conditional choice probabilities (CCPs)-weighted transition matrices and expected CCPs-weighted payoffs. \cite{aguirregabiria2007sequential} proposed a recursive algorithm that implements the Hotz-Miller estimators in each iteration.\footnote{Empirical researchers have used dynamic games to model and analyze models with \emph{infinite dependence}, including oligopolistic competition \cite{estimation_jeziorski_2014, unobserved_igami_2016, maican2018entry, salz2020estimating}, dynamic pricing decisions \citep{ellickson2012repositioning}, salesforce compensation dynamics \citep{bonuses_chung_2014, misra2011structural}, migration \citep{bishop2012dynamic, coate2013parental, ma2019learning},  housing market outcomes \citep{economics_kuminoff_2013, dynamic_murphy_2018}, school matching outcomes \citep{cupids_galichon_2015}, preschool education \cite{intergenerational_heckman_2016}, female labor supply \cite{Altuug1998, gayle2012, gayle2015accounts}, mammography choice \cite{Fang2015}, and kidney donations \cite{agarwal2021equilibrium}.}

In the CCP estimator of \cite{Hotz1993} that do not rely on finite dependence, the costs of inversion of the matrix or solving value functions recursively through forward simulation grow significantly with the size of the state space. When the DDC models exhibit finite dependence,  computing the value differences between decisions only requires a summation of the payoffs over \textit{finite} periods. Consequently, the CCP estimator that exploit finite dependence do not require an inversion of large matrices, leading to  substantial computational gains.\footnote{\cite{arcidiacono2011conditional} proposes an estimation based on the EM algorithm to address unobservable heterogeneity in the DDC models that exhibits finite dependence. }

A significant body of empirical research has applied the concept of finite dependence for estimation, often for a model with a terminal or absorbing state, as seen in studies across various fields such as fertility and labor supply, migration, stock market engagement, agricultural land use, smoking habits, education, occupational choice, housing decisions, and trade \citep{joensen2009academic,Altuug1998,gayle2012,bishop2012dynamic,coate2013parental,ma2019learning,ransom2022labor,khorunzhina2013structural,dynamic_scott_2014,matsumoto2014lighting,arcidiacono2016estimation,khorunzhina2017american,traiberman19aer}.  \cite{declercq2018enrollment}, \cite{mazur2014can}, and \cite{Beauchamp2015} utilize one-period finite dependence to estimate dynamic games.  These studies underscore the practicality of exploiting finite dependence in estimations, circumventing the need for computing value functions within nested fixed-point algorithms or inverting large matrices aligned with the state space. 

A previously known class of models that exhibit finite dependence is limited, however. \cite{arcidiacono2011conditional} discuss the sufficient condition for a model to display the \emph{finite dependence}, which arises when there is a terminal or absorbing state. 
\cite{arcidiacono_miller_2019nonstationary} generalize the definition of finite dependence by introducing the weights across the conditional value functions of all alternatives and  develops an algorithm to calculate the finite dependence paths. Although their approach is valuable for determining whether models exhibit finite dependence, their algorithm is more computationally costly compared to our proposed method. Additionally, the range of models that can be demonstrated to exhibit finite dependence within a few periods using their technique remains limited.


In this paper, extending the work of \cite{arcidiacono2011conditional,arcidiacono_miller_2019nonstationary}, we further generalize the definition of finite dependence while recasting the property of finite dependence as a problem of sequential weight determination. Most importantly, we show that numerous models display 2-period finite dependence under this generalized finite dependence, including a set of models that had been previously viewed as non-finite dependent due to stochastic state transitions given choices. Consequently, we may apply our computationally efficient CCP estimation method to a variety of the DDC models. 
We demonstrate, via Monte Carlo simulations, that our approach has a large computational gain over the CCP estimator based on the Hotz-Miller inversion.

\cite{arcidiacono_miller_2019nonstationary} highlight that the ex ante value function can be represented as a weighted average of the conditional value functions for all alternatives, augmented by a function tied to the conditional choice probabilities, with the total of all weights equating to one. Leveraging this concept, they suggest a sequential algorithm for identifying weights that could achieve finite dependence, though their method determines future weights solely based on future states, disregarding the current state. We expand upon this by proposing a broader definition of finite dependence, allowing weights to depending on both future and current states. Our approach introduces more flexibility in attaining finite dependence, enabling the identification of weights that maintain the future state's probability distribution unchanged after a small number of periods.



To determine the weights, we show that the existence of a weight vector resulting in a zero norm for the weighted forward transition matrix implies finite dependence in the model. Utilizing this insight, we develop a systematic approach for identifying such weight vectors, providing computational evidence that the norm of the weighted forward transition matrix becomes nears zero when our approach is applied one period forward. Mostly importantly, we show that a model displays 2-period finite dependence if the stochastic state transition is influenced by the chosen action and the state space is discretized. Given that many structural dynamic discrete choice models in economics share this characteristic, our estimation method is broadly applicable to a wide set of empirical economic models.


 We also introduce a straightforward computational algorithm for weight discovery that takes advantage of the Kronecker product structure present in state transitions. Our algorithm demonstrates a significant reduction in computational time compared to the algorithm presented by \cite{arcidiacono_miller_2019nonstationary}. Leveraging the estimated weights, we formulate CCP estimators akin to the finite-dependence estimator outlined by \cite{arcidiacono2011conditional}. 
 
 Our proposed estimation method under finite dependence substantially cuts down on computational time compared to the traditional CCP estimator based on Hotz-Miller's matrix inversion approach. This efficiency gain stems from the finite dependence approach requiring the calculation of utility differences for only a limited number of future periods, as opposed to an infinitely  many periods.  As emphasized by \cite{arcidiacono_miller_2019nonstationary}, our estimator also offers distinct advantages in non-stationary environments by avoiding the need to compute future value functions, a typically challenging task when stationarity is absent.

The remainder of this paper is structured as follows.
In Section \ref{sec:2_model}, we detail the dynamic discrete choice model and the framework that underpins its estimation. Section \ref{sec:3_weight} introduces a methodological approach for the identification of optimal weight vectors. Section \ref{sec:4_estimator} is dedicated to the exposition of the estimator, which is derived under specific characterizations.
The results of the Monte Carlo simulations are presented in Section \ref{sec:5_MCsimulation}.
Lastly, Section \ref{sec:6_game} extends the presented results to the context of a dynamic game.

\section{Dynamic discrete choice model}\label{sec:2_model}
\subsection{Baseline model}\label{sec:dynamic_discrete_choice_problem} 
  
An agent faces a decision-making problem where they must choose one of $D+1$ discrete, mutually exclusive actions each period, denoted as $d_t\in \mathcal{D} := \{0,1,2,\ldots, D\}$. The goal of the agent is to maximize their expected discounted sum of utilities over a time horizon $T$, which can either be finite or infinite. This objective is captured by the expression:

$$
\mathbbm{E} \left[ \sum_{s=0}^{T}\beta^{s} \left[ u_{t+s}(x_{t+s}, d_{t+s}) + \epsilon_{t+s}(d_{t+s}) \right] \mid d_{t},x_{t}\right],
$$
where $\beta \in (0,1)$ is the discount factor that weighs future utilities, $u_t(x_t, d_t)$ is the deterministic component of the time-separable payoff at time $t$, and $\epsilon_t(d_t)$ is a choice-specific idiosyncratic shock, reflecting the stochastic elements of the payoff. The state variable at time \(t\) is denoted by \(s_t = (x_t, \boldsymbol{\epsilon}_t)\), where \(x_t\) is a state variable observable by a researcher, while \(\boldsymbol{\epsilon}_t\) is a state variable that is unobservable to the researcher. The expectation \(\mathbb{E}\) is taken over the future values of state variables and shocks $\{x_{t+\tau}, \boldsymbol{\epsilon}_{t+\tau} \}_{\tau=1}^T$. The agent's decision at each period is influenced by both the current state \(s_t\) and the anticipated future states and shocks.

We adopt the following conventional assumptions in  dynamic discrete choice modeling \cite[c.f.][]{Rust1987, Rust1994, magnac2002identifying}: 

\begin{assumption}[Observed State Variables]\label{assumption:discrete_support}
The observed state variables $x$ have a discrete and finite support, defined as $\mathcal{X} = \{ 1,\ldots, X\}$, where $X$ is a finite cardinal number.
\end{assumption}

\begin{assumption}[Conditional Independence]\label{assumption:conditional_independence}
The transition probability density function $p_t(s_{t+1}|d_t,s_t)$ decomposes into $p_t(s_{t+1}|d_t,s_t) = g_t(\epsilon_{t+1}|x_{t+1}) f_t(x_{t+1} |d_t, x_t)$. This implies that the future state variables' transition is independent of current unobserved state given the current action and observed state. Moreover, the distribution of unobserved state variables is i.i.d. over time.
\end{assumption}

\begin{assumption}[Additive Separability]\label{assumption:additive_seperable}
The unobserved state variables are additively and separately integrated into the utility function: $U_t(d_t,x_t, \epsilon_t) = u_t(x_t,d_t) + \epsilon_t(d_t)$, where $\epsilon_t(d)$ denotes the component of the unobserved state vector $\epsilon_t = \{ \epsilon_t(d): d \in \mathcal{D} \}$ corresponding to decision $d$.
\end{assumption}


Under Assumption \ref{assumption:conditional_independence} to Assumption \ref{assumption:additive_seperable}, we define the integrated value function $  V(x_t)$ by integrating the optimal choice of the agents under the Bellman equations over the private information state variable:
\begin{equation}\label{eq:bellman}
	  V_t(x_t) = \int \max_{d \in \mathcal{D}} \left\{   u_t (d,x_t) + \epsilon_t(d) + \beta \sum_{x_{t+1} \in \mathcal{X}}   V_t(x_{t+1}) f(x_{t+1}|d, x_t)\right\} g_t(\epsilon_t|x_t)d\epsilon_t.
\end{equation}
Under the above assumptions, the integrated value function $\mathbf V := \{  V_t(x): x \in \mathcal{X} \} \in \mathcal{B}_V \subset \R^{X}$ is a vector of size $X$. 

We define the choice-specific value function $  v_t(z_t,d_t)$ as the discounted sum of future values of choosing alternative $d_t$: 
\begin{equation}
	  v_t(z_t,d_t) :=   u_t(z_t,d_t) + \beta \sum_{x_{t+1} \in \mathcal{X}}   V_t(x_{t+1}) f_t(x_{t+1}|d_t, x_t).
\end{equation}

Taking the choice option $d = 0$ as the baseline option, we define the value differences as $\tilde{  v}_t(z_t,d_t) =   v_t(z_t,d_t) -    v_t(z_t,d_t)$ and collect them as $\tilde{\boldsymbol  v}_t(x_t) = \{  {v}_t(d,x_t): d \in \mathcal{D} \backslash \{ 0\} \} \in \R^{D}$.
The conditional choice probabilities of the action $d_t$ given the state $x_t$  are defined as
\begin{equation}\begin{split}
		p_t(d_t, x_t) & =
		\boldsymbol {\Lambda}_t(\tilde{\boldsymbol  v}_t(x_t))(d_t,x_t) \\ 
		& = \int \mathbf{1} \{ \tilde {  v}_t(z_t,d_t) + \epsilon_t(d)  \ge \tilde {  v}_t(d',x_t) + \epsilon(d')  ~ \forall ~ d' \in \mathcal{D} \} d G_t(\epsilon)\quad\text{for all d' }.
	\end{split}
	\label{eq:ocp_mapping}
\end{equation}
We refer to $\boldsymbol {\Lambda}(\tilde{\boldsymbol  v}_t(x_t))(d_t,x_t)$ as the optimal conditional choice probability (OCP) mapping. 
We collect the leave-one-out conditional choice probabilities into a vector as $\boldsymbol  p_t(x_t) = \{ p_t(d,x_t) :  d \in \mathcal{D} \backslash \{ 0\} \} \in \R^{D}$.
If the agent $i$ makes an optimal choice for each current state $(x_t, \epsilon_t)$ and all future states, then $\boldsymbol  p_t(x_t) = \boldsymbol {\Lambda}_t(\tilde{\boldsymbol  v}_t(x_t))$.

Following the result of \citet[][Proposition 1]{Hotz1993}, we introduce the inversion of the OCP mapping. 
\begin{proposition}[Hotz-Miller Inversion]\label{prop:HM} Under Assumption \ref{assumption:conditional_independence} - \ref{assumption:additive_seperable},
	for any vector of differences in choice-dependent value functions $\tilde{\boldsymbol  v}_t(x) \in \R^{D}$, the OCP mapping is invertible, such that there is a one-to-one relationship between the vector of value differences and the vector of conditional choice probabilities, that is, $\tilde{\boldsymbol  v}_t(x) =\boldsymbol {\Lambda}_t^{-1}(\boldsymbol  p_t(x))$.
\end{proposition}

\subsection{Finite dependence property}\label{sec:finite_dependence}
 
We now derive a representation of the choice-specific value functions based on the analysis of \cite{arcidiacono2011conditional} and \cite{arcidiacono_miller_2019nonstationary}. The following theorem shows that  the difference between the value function $V_t(x_t)$ and the conditional value function $  v_t(x_t, d_t)$ is expressed as a function of the conditional choice probabilities. 
\begin{theorem}[Lemma 1 of \cite{arcidiacono2011conditional}]\label{thm:AM-lemma1}
There exists a real-valued function $\psi_d(\boldsymbol  p)$ for every  $d\in \mathcal{D}$ such that
\begin{equation}\label{eq:psi}
 V_t(x_t)  = \psi_{d}(\boldsymbol  p_t(x_t)) +  v_t(x_t,d).
\end{equation}
\end{theorem}  

The function $\psi_d(\boldsymbol  p_t)$ has an analytical expression or is straightforward to evaluate when $\epsilon_t$ follows the Generalized Extreme Value (GEV) distribution, of which special cases include T1EV distribution and nested logit \citep[Section 3.3 of][]{arcidiacono2011conditional}.  Moreover, the mass transport approach by \cite{chiong2016duality} can be used to evaluate  $\psi_d$ for general distributions beyond the GEV. 

For each triple $(x_t,x_{t+1},d_t)\in \mathcal{X}^2\times\mathcal{D}$, consider a  vector of weights $$\mathbf{w}_{t+1}(x_{t+1}|x_t, d_t):= ( \mathrm{w}_{t+1}(x_{t+1}, 0|x_t, d_t),...,\mathrm{w}_{t+1}(x_{t+1},D|x_t, d_t) )\t$$ that satisfy $\sum_{d'\in \mathcal{D}} \mathrm{w}_{t+1}(x_{t+1}, d'|x_t, d_t) =1$ and $|| \mathbf{w}_{t+1}(x_{t+1}|x_t, d_t)||<\infty$. Evaluating equation (\ref{eq:psi}) at  $t+1$, substituting equation (\ref{eq:value_function}) at $t+1$ into the resulting equation, and taking their weighted averages across $D$ choices give   
\begin{align}\label{eq:weighted}
    V_{t+1}(x_{t+1})&=  \sum_{d' \in \mathcal{D}} \mathrm{w}_{t+1}(x_{t+1}, d'|x_t, d_t)[  u_{t+1}(x_{t+1},d') +\psi_{d'}(\boldsymbol  p_{t+1}(x_{t+1}))] \nonumber\\
    &\quad + \beta \sum_{x_{t+2} \in \mathcal{X}}   V_{t+2}(x_{t+2}) \sum_{d' \in \mathcal{D}} \mathrm{w}_{t+1}(x_{t+1}, d'|x_t, d_t) f_{t+1}(x_{t+2}| x_{t+1},d').
\end{align}    
 
By substituting equation (\ref{eq:weighted})  into the right-hand side of equation (\ref{eq:value_function}), we obtain the following representation of choice-specific value functions:
\begin{align}
  v_t(x_t, d_t) =&   u_t(x_t, d_t)  + \beta \sum_{x_{t+1}\in\mathcal{X}} \left[ \bar u^{\mathbf{w}_{t+1} } _{t+1}(x_{t+1}) +\bar \psi^{\mathbf{w}_{t+1} }(\boldsymbol  p_{t+1}(x_{t+1}))\right]f_t(x_{t+1}|x_t, d_t)\nonumber \\
& + \beta^2 \sum_{ x_{t+2}\in\mathcal{X}}  V_{t+2}(x_{t+2})\sum_{ x_{t+1}\in\mathcal{X}} \bar f_{t+1}^{\mathbf{w}_{t+1} }(x_{t+2}|x_{t+1}) f_t(x_{t+1}|x_t, d_t),\label{eq:representation}
\end{align}
where 
\begin{align*}
& \bar  u^{\mathbf{w}_{t+1} } _{t+1}(x_{t+1}):= \sum_{d' \in \mathcal{D}}\mathrm{w}_{t+1}(x_{t+1}, d'|x_t, d_t)    u_{t+1}(x_{t+1},d'),\\
&\bar\psi^{\mathbf{w}_{t+1} }(\boldsymbol  p_{t+1}(x_{t+1})):= \sum_{d' \in \mathcal{D}} \mathrm{w}_{t+1}(x_{t+1}, d'|x_t, d_t) \psi_{d'}(\boldsymbol  p_{t+1}(x_{t+1})),\quad\text{and} \\
& \quad \bar f_{t+1}^{\mathbf{w}_{t+1}}(x_{t+2}|x_{t+1}):=\sum_{d' \in \mathcal{D}} \mathrm{w}_{t+1}(x_{t+1}, d'|x_t, d_t) f_{t+1}(x_{t+2}|x_{t+1},d').
\end{align*}
 
The formulation presented in (\ref{eq:representation}) is derived from Theorem 1, equation (2.6), within the work of \cite{arcidiacono_miller_2019nonstationary}. A critical distinction in our approach lies in the selection of decision weights, $\mathbf{w}_{t+1}$. Unlike the methodology proposed by Arcidiacono and Miller, who define weights solely based on the pair $(x_{t+1},d_t)$, we determine weights for each combination of $(x_t,x_{t+1},d_t)$. This extension modifies and broadens their original concept of 1-period dependence.



Taking a choice of $0$ as the baseline choice, let the difference in value functions between two distinct choices, $d_t$ and $0$, be represented by $\tilde{v}_t(x_t, d_t) \equiv v_t(x_t, d_t) - v_t(x_t, 0)$ for $d_t\in \mathcal{D}\backslash \{0\}$. From (\ref{eq:representation}), this difference can be expressed as follows: 
\begin{align} \label{eq:value_function}
\tilde{v}_t(x_t, d_t) = & \ \tilde{u}_t(x_t, d_t) + \beta \sum_{x_{t+1}\in\mathcal{X}} \left[ \bar{u}^{\mathbf{w}_{t+1}}_{t+1}(x_{t+1}) + \bar{\psi}^{\mathbf{w}_{t+1}}(\boldsymbol{p}_{t+1}(x_{t+1})) \right] \tilde{f}_t(x_{t+1}|x_t, d_t) \nonumber \\
& + \beta^2 \sum_{x_{t+2}\in\mathcal{X}} V_{t+2}(x_{t+2}) \left[ \sum_{x_{t+1}\in\mathcal{X}} \bar{f}_{t+1}^{\mathbf{w}_{t+1}}(x_{t+2}|x_{t+1}) \tilde{f}_t(x_{t+1}|x_t, d_t) \right], 
\end{align} 
where $\tilde{u}_t(x_t, d_t) = u_t(x_t, d_t) - u_t(x_t, 0)$ captures the immediate utility difference between choosing $d_t$ over the choice of $0$ while $\tilde{f}_t(x_{t+1}|x_t, d_t) = f_t(x_{t+1}|x_t, d_t) - f_t(x_{t+1}|x_t, 0)$ is the difference in transition probabilities between the choices. The term $\bar{u}^{\mathbf{w}_{t+1}}_{t+1}(x_{t+1}) + \bar{\psi}^{\mathbf{w}_{t+1}}(\boldsymbol{p}_{t+1}(x_{t+1}))$ represents the expected future utilities in period $t+1$ and $V_{t+2}(x_{t+2})$ represents the expected value function in time $t+2$.  

\begin{definition}[1-period dependence]
A model is said to exhibit 1-period dependence when for every sequence of states and decisions $(x_t, x_{t+2}, d_t) \in \mathcal{X}^2 \times \mathcal{D}$, there exists a corresponding set of decision weights $\{\mathbf{w}_{t+1}(x_{t+1} | x_t, d_t): (x_{t+1} | x_t, d_t) \in \mathcal{X}^2 \times \mathcal{D}\}$ that satisfy the following condition:

\begin{align}\label{eq:2-dependence}
\sum_{x_{t+1} \in \mathcal{X}} \bar f_{t+1}^{\mathbf{w}_{t+1}}(x_{t+2} | x_{t+1})
[f_t(x_{t+1} | x_t, d_t) - f_t(x_{t+1} | x_t, 0)] = 0.
\end{align}
\end{definition}

Our definition of 1-period dependence broadens the definition established by \cite{arcidiacono_miller_2019nonstationary} due to the inclusion of $x_t$ in the decision weights, allowing for a more expansive class of models to exhibit 1-period dependence.  Taking advantage of the flexibility arising from the inclusion of $x_t$, we develop a numerical method designed to precisely identify the decision weights. This method facilitates the construction of models with finite dependence, ensuring that our approach is practically applicable in a variety of models beyond a model with terminal states or renewal choices.


Under 1-period dependence, value differences in choice-specific value functions between choices do not depend on the integrated value function $V_{t+2}$ as
\begin{equation}\label{eq:value-differences}
\tilde{  v}_t(x_t, d_t) =   \tilde u_t(x_t, d_t)  + \beta \sum_{x_{t+1}\in\mathcal{X}} \left[ \bar u^{\mathbf{w}_{t+1}} _{t+1}(x_{t+1}) +\bar \psi^{\mathbf{w}_{t+1}}(\boldsymbol  p_{t+1}(x_{t+1}))\right]\tilde  f_t(x_{t+1}|x_t, d_t).
\end{equation}

This analysis can be extended to finite dependence over multiple periods.  
Given a triple $(x_{t+\tau} | x_t, d_t) \in \mathcal{X}^{2} \times \mathcal{D}$ for $\tau=1,\dots,\mathcal{T}$, define the vector of weights $\mathbf{w}_{t+\tau}(x_{t+\tau} | x_t, d_t)$ as
\begin{equation}
\mathbf{w}_{t+\tau}(x_{t+\tau} | x_t, d_t) := \left( \mathrm{w}_{{t+\tau}}(x_{t+\tau}, 0| x_t, d_t), \ldots, \mathrm{w}_{{t+\tau}}(x_{t+\tau}, D| x_t, d_t) \right)^\top,
\end{equation}
which satisfy the normalization condition \(\sum_{d' \in \mathcal{D}} w_{{t+\tau}}(x_{t+\tau},d' | x_t, d_t) = 1, \) and the boundedness condition  \( \left\|\mathbf{w}_{{t+\tau}}\right\| < \infty.\)

Let $\mathbf{W}_{t+\tau}(x_t, \ldots, x_{t+\tau}, d_t)$ be defined as
\begin{equation}\label{eq:W}
\mathbf{W}_{t+\tau}(x_t,\ldots,x_{t+\tau},d_t) := \left(\mathbf{w}_{t}(x_t | x_t, d_t), \ldots, \mathbf{w}_{t+\tau}(x_{t+\tau} | x_t, d_t)\right).
\end{equation}
Then, define $\kappa_{t+\tau}^{\mathbf{W}_{t+\tau}}(x_{t+\tau+1}|x_t, d_t)$ recursively as
\begin{equation}\label{eq:kappa_recursive}
\begin{aligned}
\kappa_{t+\tau}^{\mathbf{W}_{t+\tau}}(x_{t+\tau+1} | x_t, d_t) := \left\{
\begin{array}{ll}
f_t(x_{t+1} | x_t, d_t), & \text{for } \tau = 0, \\
\sum_{x_{t+\tau} \in \mathcal{X}} \bar{f}_{t+\tau}^{\mathbf{w}_{t+\tau}}(x_{t+\tau+1} | x_{t+\tau}, x_t, d_t ) \\
\quad \times \kappa_{t+\tau-1}^{\mathbf{W}_{t+\tau-1}}(x_{t+\tau} | x_t, d_t), & \text{for } \tau = 1, \ldots, \mathcal{T},
\end{array}
\right.
\end{aligned}
\end{equation}
where the function $\bar{f}_{t+\tau}^{\mathbf{w}_{t+\tau}}$ uses the weights $\mathbf{w}_{t+\tau}$ to adjust the transition probabilities appropriately: \( \bar{f}_{t+\tau}^{\mathbf{w}_{t+\tau}}(x_{t+\tau+1} | x_{t+\tau}, x_t, d_t ) = \sum_{d \in \mathcal{D}}f_{t+\tau}(x_{t+\tau+1} | x_{t+\tau}, d) \mathrm{w}_{t+\tau} (x_{t+\tau+1},d | x_{t}, d_t) \).

Then, we define the $\rho$-period finite dependence as follows.
\begin{definition}[$\rho$-period finite dependence]
Let $\rho \geq 0$ be an integer. A model exhibits $\rho$-period finite dependence if, for all $(x_t, x_{t+\rho+1}, d_t) \in \mathcal{X}^2 \times \mathcal{D}$, there exists a set of decision weights $\mathbf{W}_{t+\rho}(x_t, \ldots, x_{t+\rho}, d_t)$ as defined by equation (\ref{eq:W}), such that the following condition holds:
\begin{equation}\label{eq:rho-dependence}
\kappa_{t+\rho}^{\mathbf{W}_{t+\rho}}(x_{t+\rho+1} | x_t, d_t) - \kappa_{t+\rho}^{\mathbf{W}_{t+\rho}}(x_{t+\rho+1} | x_t, 0) = 0.
\end{equation}
\end{definition}

Applying induction to equation (\ref{eq:weighted}) together with equation (\ref{eq:value_function}), we obtain the following representation of value differences under $\rho$-period finite dependence: 
\begin{align} \label{eq:rho-dependence_value}
\tilde{  v}_t(x_t, d_t) =   \tilde u_t(x_t, d_t)  + \sum_{\tau=1}^{\rho}\beta^{\tau}  \sum_{x_{t+\tau}\in\mathcal{X}} \left[ \bar u^{\mathbf{w}_{t+\tau}} _{t+\tau}(x_{t+\tau}) +\bar \psi^{\mathbf{w}_{t+\tau}}(\boldsymbol  p_{t+\tau}(x_{t+\tau}))\right]\tilde\kappa_{t+\tau-1}^{\mathbf{W}_{t+\tau-1}}(x_{t+\tau}|x_t, d_t),
\end{align} 
where
$\tilde\kappa_{t+\tau-1}^{\mathbf{W}_{t+\tau-1}}(x_{t+\tau}|x_t, d_t):=
\kappa_{t+\tau-1}^{\mathbf{W}_{t+\tau-1}}(x_{t+\tau}|x_t, d_t)-
    \kappa_{t+\tau-1}^{\mathbf{W}_{t+\tau-1}}(x_{t+\tau}|x_t,0)
$. By substituting the derived representation of value differences in Equation (\ref{eq:rho-dependence_value}) into the right hand side of the OCP mapping in Equation (\ref{eq:ocp_mapping}),  
we can derive the likelihood function for estimating the model parameters, as discussed in Section \ref{sec:4_estimator}.

An important special case is the 2-period finite dependence. A model exhibits $2$-period finite dependence if there exists a set of decision weights  $\{\mathbf{w}_{t+1}(x_{t+1} | x_t, d_t): (x_{t+1},x_t,d_t)\in \mathcal{X}^2\in \mathcal{D}\}$ and $\{\mathbf{w}_{t+2}(x_{t+2} | x_t, d_t)(x_{t+1},x_t,d_t)\in \mathcal{X}^2\in \mathcal{D}\}$ that satisfy the following condition:
$$ 
\sum_{x_{t+2} \in \mathcal{X}} \bar f_{t+2}^{\mathbf{w}_{t+2}}(x_{t+3} | x_{t+2})\sum_{x_{t+1} \in \mathcal{X}}\bar f_{t+1}^{\mathbf{w}_{t+1}}(x_{t+2} | x_{t+1})
[f_t(x_{t+2} | x_{t+1}, d_t) - f_t(x_{t+2} | x_{t+1}, 0)] = 0
$$
for all $(x_t,x_{t+3},d_{t})\in\mathcal{X}^2\times\mathcal{D}$. We introduce an algorithm that identifies such weights and establish conditions under which a model exhibits 2-period finite dependence.

\section{Solving the decision weights}\label{sec:3_weight}

The predominant approach employed to determine the existence of finite dependence in the current body of research is the ``guess and verify" method. As a consequence, nearly all empirical applications that invoke the concept of finite dependence have concentrated on two specific cases of one-period dependence: models that incorporate a terminal choice or a renewal choice \citep{arcidiacono_miller_2019nonstationary}.

To broaden the scope of the application, we propose a computational method for determining decision weights that satisfy the conditions for the finite dependence property, as defined by equations (\ref{eq:2-dependence}) and (\ref{eq:rho-dependence}). Utilizing these weights, we develop an estimator for conditional choice probabilities predicated on the concept of near-finite dependence within nonstationary models.

To illustrate the idea, we restrict our analysis to a model with binary choices, denoted by $\mathcal{D} = \{0,1\}$, in this section. Given the dichotomous nature of decision-making, we adopt the streamlined notation for the weight function given by $\mathrm{w}_{t+\tau} (x_{t+\tau} | x_t) := \mathrm{w}_{t+\tau} (x_{t+\tau}, 1 | x_t, 1)$ so that $\mathrm{w}_{t+\tau} (x_{t+\tau}, 0 | x_t, 1)$ is represented by $1-\mathrm{w}_{t+\tau} (x_{t+\tau} | x_t)$. Our approach is straightforwardly applicable to a model with three or more choices. 


Let $\mathbf{F}_{0,\tau}$ and $\mathbf{F}_{1,\tau}$ be Markovian transition matrices representing the probabilities of transitioning between states when action $0$ and action $1$ are chosen at time $\tau$, respectively. The matrices are defined by the transition probabilities $f_t(x_{\tau+1}^{(i)} | x_{\tau}^{(j)}, a)$ for all states $i, j$ and action $a$ as:

\[
\mathbf{F}_{0,\tau} := \left[\begin{array}{ccc}
    f_t(x_{\tau+1}^{(1)} | x_{\tau}^{(1)}, 0) & \cdots & f_t(x_{\tau+1}^{(X)} | x_{\tau}^{(1)}, 0) \\
    \vdots & \ddots & \vdots \\
    f_t(x_{\tau+1}^{(1)} | x_{\tau}^{(X)}, 0) & \cdots & f_t(x_{\tau+1}^{(X)} | x_{\tau}^{(X)}, 0)
\end{array}\right]\quad\text{and}
\]

\[
\mathbf{F}_{1,\tau} := \left[\begin{array}{ccc}
    f_t(x_{\tau+1}^{(1)} | x_{\tau}^{(1)}, 1) & \cdots & f_t(x_{\tau+1}^{(X)} | x_{\tau}^{(1)}, 1) \\
    \vdots & \ddots & \vdots \\
    f_t(x_{\tau+1}^{(1)} | x_{\tau}^{(X)}, 1) & \cdots & f_t(x_{\tau+1}^{(X)} | x_{\tau}^{(X)}, 1)
\end{array}\right].
\]
Further, we define the difference of transition density as  \[ \tilde {\mathbf{F}}_{\tau} := \mathbf{F}_{1,\tau} -  \mathbf{F}_{0,\tau} =\left[\begin{array}{ccc}
    \tilde{f}_t(x_{\tau+1}^{(1)} | x_{\tau}^{(1)}, 1) & \cdots & \tilde{f}_t(x_{\tau+1}^{(X)} | x_{\tau}^{(1)}, 1) \\
    \vdots & \ddots & \vdots \\
    \tilde{f}_t(x_{\tau+1}^{(1)} | x_{\tau}^{(X)}, 1) & \cdots & \tilde{f}_t(x_{\tau+1}^{(X)} | x_{\tau}^{(X)}, 1)
\end{array}\right]. \]

\subsection{1-period dependence}
 
For 1-period finite dependence, we propose a method of computing a weight vector $\mathbf{w}_{t+1}:=\{\mathrm{w}_{t+1}(x_{t+1}|x_{t}): (x_t,x_{t+1})\in\mathcal{X}^2\}$ such that the difference between $v(x_t, 1)$ and $v(x_t, 0)$ does not reflect the influence of payoffs beyond time $t+1$. The weight vector for 1-period finite dependence must satisfy the following condition:
\begin{equation} \label{eq:2-dependence-2}
\sum_{x_{t+1}} \tilde {f}_{t+1}(x_{t+2} | x_{t+1},1 ) \mathrm{w}_{t+1} (x_{t+1} | x_{t}) \tilde{f}_t(x_{t+1} | x_t,1) +\sum_{x_{t+1}} \tilde{f}_t(x_{t+1} | x_t,1) {f}_{t+1}(x_{t+2} | x_{t+1},0 )  = 0,  
\end{equation}  
for all $(x_t,x_{t+2})\in\mathcal{X}^2$.

Let $\mathrm{\check{w}}_{t+1}(x_{t+1}|x_t) := \tilde{f}(x_{t+1} | x_t, 1) \mathrm{w}_{t+1}(x_{t+1}|x_t)$. Then,
rewriting (\ref{eq:2-dependence-2}) for all values of $(x_t,x_{t+2})\in\mathcal{X}^2$ in matrices, we may express (\ref{eq:2-dependence-2}) as $$ \mathbf{\check{W}}_{t+1} \mathbf{\tilde{F}}_{t+1} + \mathbf{\tilde{F}}_{t} \mathbf{F}_{0, t+1}=0,$$ where
$$\mathbf{ \check{W}}_{t+1} = \begin{bmatrix}
     \tilde{f}(x_{t+1}^{(1)} | x_t^{(1)}, 1) \mathrm{w}_{t+1}(x_{t+1}^{(1)} | x_t^{(1)}) & \cdots &  \tilde{f}(x_{t+1}^{(X)} | x_t^{(1)}, 1) \mathrm{w}_{t+1}(x_{t+1}^{(X)} | x_t^{(1)}) \\
     \vdots & \ddots & \vdots \\
     \tilde{f}(x_{t+1}^{(1)} | x_t^{(X)}, 1) \mathrm{w}_{t+1}(x_{t+1}^{(1)} | x_t^{(X)}) & \cdots &  \tilde{f}(x_{t+1}^{(X)} | x_t^{(X)}, 1) \mathrm{w}_{t+1}(x_{t+1}^{(X)}|x_t^{(X)})
\end{bmatrix}.$$ 





We define a new weight function $\mathrm{\check{w}}_{t+1}(x_{t+1}|x_t) = \tilde{f}(x_{t+1} | x_t) \mathrm{w}_{t+1}(x_{t+1}|x_t)$. This function $\mathbf{\check{w}}$ solves the linear system as described by the equation:
\begin{equation}\label{eq:weight_minimization}
\mathbf{\check{w}}_{t+1} = -\mathbf{\tilde{F}} \mathbf{F}_0 (\mathbf{\tilde{F}})^+ , \quad \mathbf{\tilde{F}} \mathbf{F}(\mathbf{w}_{t+1}) = \mathbf{\check{w}} \mathbf{\tilde{F}} + \mathbf{\tilde{F}} \mathbf{F}_0.
\end{equation}
where $(\mathbf{\tilde{F}})^+$ denotes the Moore–Penrose inverse of the transposed matrix $\mathbf{\tilde{F}}$.

\subsection{$\rho$-period finite dependence}

Consider the scenario where our goal is to minimize the weights over two distinct periods, denoted by the weight vectors $\mathbf{w}_1$ and $\mathbf{w}_2$. Given the complexity inherent in estimating $|X|^3$ possible trajectories, each dependent on a unique combination of these weights, the minimization procedure is structured to proceed sequentially. We assume that the $(i,j)$-th element of $\mathbf{w}_{t+1}$ is $\mathbf{w}_{t+1}(x_{it+1}^{(j)} | x_{it}^{(i)})$ and that of $\mathbf{w}_{t+2}$ is $\mathbf{w}_2(x_{it+2}^{(j)} | x_{it}^{(i)})$.

We initiate the optimization procedure with a given arbitrary weight vector $\mathbf{w}_{t+1}$. Our objective is to mitigate the influence of future expected values on the likelihood of the current decision. To achieve this, we select a second-period weight vector $\mathbf{w}_{t+2}$ by solving the following optimization problem:

\[
\mathbf{w}_{t+2} = \arg \min_{\mathbf{w}} \mathbf{\tilde{F}} \mathbf{F}(\mathbf{w}_{t+1}) \mathbf{F}(\mathbf{w}),
\]
where $\mathbf{\tilde{F}}$ symbolizes the initial configuration of the system, and $\mathbf{F}(\mathbf{w})$ denotes the transformation matrix parameterized by weight vector $\mathbf{w}$.

Analogous to the methodology applied for deducing $\mathbf{w}_{t+1}$, we determine $\mathbf{w}_{t+2}$ by employing a similar computational strategy, which facilitates the derivation of an explicit expression for $\mathbf{\check{w}}_{t+2}$. The relationship between $\mathbf{\check{w}}_{t+2}$ and $\mathbf{w}_{t+2}$ is established as follows:

\[
\mathbf{\check{w}}_{t+2}(x_{t+2} | x_{t}) = \tilde{\kappa}_{t+1}^{\mathbf{W}_{t+1}}(x_{t+2}|x_{t},d_t) \mathbf{w}_{t+2}(x_{t+2} | x_{t}),
\]
where the definition of $\tilde{\kappa}_{t+1}^{\mathbf{W}_{t+1}}(x_{t+2}|x_{t},d_t)$  follows \eqref{eq:rho-dependence} and \eqref{eq:rho-dependence_value}. 

The definition results in the following equation:

\[
\begin{split}
\mathbf{\check{w}}_{t+2} &= - \left( \mathbf{\tilde{F}} \mathbf{F}(\mathbf{w}_{t+1}) \right) \left( \mathbf{F}_0(\mathbf{\tilde{F}}) \right)^+, \\
\mathbf{\tilde{F}} \mathbf{F}(\mathbf{w}_{t+1}) \mathbf{F}(\mathbf{w}_{t+2}) &= \mathbf{\check{w}}_{t+2} \mathbf{\tilde{F}} + \mathbf{\tilde{F}} \mathbf{F}(\mathbf{w}_{t+1}) \mathbf{F}_0 \\
&= \mathbf{\check{w}}_{t+2} \mathbf{\tilde{F}} + \mathbf{\check{w}}^{(1)} \mathbf{\tilde{F}}  \mathbf{F}_0 + \mathbf{\tilde{F}} (\mathbf{F}_0)^2.
\end{split}
\]

This formulation encapsulates the sequential optimization of weights, accounting for the intertemporal dependencies inherent in the decision-making process.

We can further extend the results to $\rho$ period, where 
\[\begin{split}
& \mathbf{\check w}_{t+\rho} = - \big( \mathbf{\tilde F} \underbrace{ \mathbf{F}(\mathbf{w}_{t+1}), \ldots \mathbf{F}(\mathbf{w}_{t+\rho-1}) }_{\rho - 1 \textbf{ periods}} \big)\mathbf{F}_0  (\mathbf{\tilde F})^+, 
\\ & 
\big( \mathbf{\tilde F} \underbrace{ \mathbf{F}(\mathbf{w}_{t+1}), \ldots \mathbf{F}(\mathbf{w}_{t+\rho-1}) }_{\rho - 1 \textbf{ periods}} \big) \mathbf{F}(\mathbf{w}_{t+\rho}) = \sum_{s=1}^{\rho} \mathbf{\check w}_{t+s} \mathbf{\tilde F} (\mathbf{F}_0)^{\rho-s} + \mathbf{\tilde F} (\mathbf{F}_0)^{\rho}).
\end{split}\]

Let $\mathbf{F}: \mathcal{W} \to \mathcal{F}$ be a function mapping from some domain $\mathcal{W}$ to a codomain $\mathcal{F}$. Let $\mathbf{\tilde F}$ be a fixed element in $\mathcal{F}$, and let $\mathbf{\tilde F}^+$ denote the Moore-Penrose pseudoinverse of $\mathbf{\tilde F}$. For a positive integer $\rho$, define the sequence $({\tilde{\mathbf F}}^{(k)})_{k=0}^{\rho}$ in $\mathcal{F}$ recursively by
\begin{equation}
\tilde{\mathbf F}^{(k)} \equiv
\begin{cases}
\mathbf{\tilde F}, & \text{if } k=0, \\
\tilde{\mathbf F}^{(k-1)} \mathbf{F}(\mathbf{w}_{t+k}), & \text{for } k=1,2,\ldots,\rho,
\end{cases}
\end{equation}
where each $\mathbf{w}_{t+k}$ is an element of $\mathcal{W}$.

Then, the transformation that constitutes the one period ahead bias correction mapping from $\tilde{\mathbf F}^{(\rho-1)}$ to $\tilde{\mathbf F}^{(\rho)}$ can be expressed as
\begin{equation}
\tilde{\mathbf F}^{(\rho)} = \tilde{\mathbf F}^{(\rho-1)} \mathbf{ F}_0 (I - \mathbf{\tilde F}^+ \mathbf{\tilde F} ),
\end{equation}
where $\mathbf{\tilde F}_0$ is a fixed element in $\mathcal{F}$ associated with the bias correction, and $I$ is the identity operator on $\mathcal{F}$.

We have, for $\rho$ period
\[\begin{split}
    \underbrace{\mathbf{\tilde F} \mathbf{F}(\mathbf{w}_{t+1}), \ldots,  \mathbf{F}(\mathbf{w}_{t+\rho})}_{\tilde{\mathbf F}^{(\rho)}} &  = \mathbf{\check w}_{t+\rho}\mathbf{\tilde F} + \tilde{\mathbf F}^{(\rho-1)} \mathbf{F}_0,  = \sum_{s=1}^{\rho} \mathbf{\check w}_{t+s}\mathbf{\tilde F} \mathbf{F}_0^{\rho - s} + \mathbf{\tilde F} (\mathbf{F}_0 )^{\rho}.
\end{split} \]

This formulation assumes that the operations involved are well-defined, such as the product of elements in $\mathcal{F}$ and the existence of the Moore-Penrose pseudoinverse. It also assumes that the sequence of elements $\mathbf{w}_{t+k}$ is given and that the recursive definition of $\tilde{\mathbf F}^{(k)}$ is well-founded. The proposition is now more rigorous as it clearly defines the context, the recursive relationship, and the transformation in question.

In order to ascertain the rate of contraction within the mapping, it is imperative to conduct an analysis of the singular values associated with the operator $\mathbf{ F}_0 (I - \mathbf{\tilde F}^+ \mathbf{\tilde F} )$. It is noteworthy to mention that the term $(I - \mathbf{\tilde F}^+ \mathbf{\tilde F} )$ represents the orthogonal projector onto the null space of the matrix $\mathbf{\tilde F}$. This orthogonal projector is instrumental in determining the behavior of the mapping in the context of convergence and stability analysis.

Let us define the projection matrix onto the null space of $\mathbf{\tilde F}$ as
\[
\mathcal{P}_{\mathbf{\tilde F}} = \left( \mathbf{I} -  \mathbf{\tilde F}^+ \mathbf{\tilde F} \right).
\]
Given the relationship \( \mathbf{\check w}_{t+\rho} = - \tilde{\mathbf F}^{(\rho-1)} \mathbf{F}_0 \mathbf{\tilde F}^+ \), we can plug in to obtain:

\[
\begin{split}
    \tilde{\mathbf F}^{(\rho)} &=  \tilde{\mathbf F}^{(\rho-1)} \mathbf{F}_0  \Big( \mathbf{I} -  \mathbf{\tilde F}^+ \mathbf{\tilde F} \Big) =  \tilde{\mathbf F}^{(1)} \Big( \mathbf{F}_0  \big( \mathbf{I} -  \mathbf{\tilde F}^+ \mathbf{\tilde F} \big)\Big)^{\rho - 1} \\
    &= \Big(\mathbf{\check w}_{t+1} \mathbf{\tilde F} + \mathbf{\tilde F} \mathbf{F}_0 \Big) \Big(\mathbf{F}_0 \big( \mathbf{I} -  \mathbf{\tilde F}^+ \mathbf{\tilde F} \big)\Big)^{\rho - 1}.
\end{split}
\]

When weights are chosen optimally, we have:

\[
\begin{split}
    \tilde{\mathbf F}^{(\rho)} & = - \big(\mathbf{\tilde F} \mathbf{F}_0 \big) \Big( \mathbf{F}_0 \mathcal{P}_{\mathbf{\tilde F}} \Big)^{\rho - 1} \Big( \mathbf{I} -  \big(\mathbf{\tilde F}\Big( \mathbf{F}_0 \mathcal{P}_{\mathbf{\tilde F}} \Big)^{\rho - 1} \big)^+ \big(\mathbf{\tilde F}\Big( \mathbf{F}_0 \mathcal{P}_{\mathbf{\tilde F}} \Big)^{\rho - 1} \big) \Big),
\end{split}
\]

and also:

\[
\mathbf{\check w}^{(\rho)*}_{t+1} = - \Big(\mathbf{\tilde F} \mathbf{F}_0 \Big(\mathbf{F}_0 \mathcal{P}_{\mathbf{\tilde F}}\Big)^{\rho - 1}\Big)  \Big(\mathbf{\tilde F} \Big(\mathbf{F}_0 \mathcal{P}_{\mathbf{\tilde F}}\Big)^{\rho - 1} \Big)^+ .
\]

\subsection{2-period finite dependence}

We now consider a model with 2-period finite dependence. To express the weight vector $\mathbf{w}^{(3)}$, we first relate it to $\mathbf{w}_{t+2}$, which is in turn expressed as a function of $\mathbf{w}_{t+1}$. If we re-optimize $\mathbf{w}_{t+1}$ given $\mathbf{w}_{t+2}$ and apply the previously derived expression
\[
\mathbf{w}_{t+2}(\mathbf{w}_{t+1}) = - \left(\mathbf{\check w}_{t+1} \mathbf{\tilde F} + \mathbf{\tilde F} \mathbf{F}_0 \right) \mathbf{F}_0  \mathbf{\tilde F}^+ ,
\]
we obtain the following:
\[
\begin{split}
    \mathbf{\tilde F} \mathbf{F}(\mathbf{w}_{t+1}) \mathbf{F}(\mathbf{w}_{t+2}(\mathbf{w}_{t+1})) & =  \left( - \left(\mathbf{\check w}_{t+1} \mathbf{\tilde F} + \mathbf{\tilde F} \mathbf{F}_0 \right) \mathbf{F}_0  \mathbf{\tilde F}^+  \mathbf{\tilde F} + \mathbf{\check w}_{t+1} \mathbf{\tilde F}  \mathbf{F}_0 + \mathbf{\tilde F} \mathbf{F}_0^2 \right)  \\
    & = \mathbf{w}_{t+1} \mathbf{\tilde F} \mathbf{F}_0  \left( \mathbf{I} -  \mathbf{\tilde F}^+  \mathbf{\tilde F} \right) + \mathbf{\tilde F} \mathbf{F}_0^2 \left( \mathbf{I} -  \mathbf{\tilde F}^+  \mathbf{\tilde F} \right).
\end{split}
\]

In terms of optimality, when determining the 1-period weight matrix, the first period weight matrix should be selected optimally as
\[
\mathbf{\check w}^{(2)*}_{t+1} = -\mathbf{\tilde F} \mathbf{F}_0^2 \mathcal{P}_{\mathbf{\tilde F}} \left(\mathbf{\tilde F} \mathbf{F}_0 \mathcal{P}_{\mathbf{\tilde F}} \right)^+ .
\]
And
\[
\tilde{\mathbf F}^{(2)} = \mathbf{\tilde F} \mathbf{F}_0 \mathbf{F}_0 \mathcal{P}_{\mathbf{\tilde F}} \left( \mathbf{I} -  (\mathbf{\tilde F} \mathbf{F}_0 \mathcal{P}_{\mathbf{\tilde F}} )^+ (\mathbf{\tilde F} \mathbf{F}_0 \mathcal{P}_{\mathbf{\tilde F}}) \right).
\]

Let $\mathbf{\tilde F} \in \mathbb{R}^{(D)X \times X}$, and consider its Singular Value Decomposition (SVD):
\[
\mathbf{\tilde F} = \mathbf{U}_{\tilde F} \mathbf{S}_{\tilde F} \mathbf{V}_{\mathbf{\tilde F}}^\top,
\]
where $\mathbf{U}_{\tilde F} \in \mathbb{R}^{(D)X \times (D)X}$ contains the left singular vectors, $\mathbf{S}_{\tilde F} \in \mathbb{R}^{(D)X \times X}$ is a diagonal matrix with the singular values, and $\mathbf{V}_{\mathbf{\tilde F}} \in \mathbb{R}^{X \times X}$ contains the right singular vectors. Since $\mathbf{\tilde F}$ does not have full rank, we can partition $\mathbf{S}_{\tilde F}$ as follows:
\[ 
\mathbf{S}_{\tilde F} = \begin{bmatrix}
        \mathbf{S}_{\tilde F,00} & \mathbf{0} \\
        \mathbf{0} & \mathbf{0} 
\end{bmatrix},
\]
where $\mathbf{S}_{\tilde F,00}$ is a $\text{rank}(\mathbf{\tilde F}) \times \text{rank}(\mathbf{\tilde F})$ matrix.

Given that $\mathbf{\tilde F}$ is the difference between two Markov transition matrices, and using the properties inherent to Markov matrices, we can infer that the rank of $\mathbf{\tilde F}$ is less than $X$. Thus, we can express the projection matrix $\mathcal{P}_{\mathbf{\tilde F}}$ onto the null space of $\mathbf{\tilde F}$ using the right singular vectors corresponding to the zero singular values as follows:
\[
\mathcal{P}_{\mathbf{\tilde F}} = \mathbf{V}_{\mathbf{\tilde F}} \Sigma_{\text{Null}(\tilde F)} \mathbf{V}_{\mathbf{\tilde F}}^\top,
\]
where $\Sigma_{\text{Null}(\tilde F)}$ is a diagonal matrix with entries
\[
\Sigma_{\text{Null}(\tilde F)} = \text{diag}\left( \underbrace{0, \ldots, 0}_{\text{rank}(\mathbf{\tilde F})\text{ zeros}}, 1, \ldots, 1 \right).
\]

Now, let $\mathbf{S}_0 = \mathbf{V}_{\mathbf{\tilde F}}^\top \mathbf{F}_0 \mathbf{V}_{\mathbf{\tilde F}}$. Since $\mathbf{V}_{\mathbf{\tilde F}}^\top$ is unitary, we have $\mathbf{F}_0 = \mathbf{V}_{\mathbf{\tilde F}} \mathbf{S}_0 \mathbf{V}_{\mathbf{\tilde F}}^\top$. Partitioning the matrix of $\mathbf{S}_0$ yields: 
\[ \mathbf{S}_0 = \begin{bmatrix}
        \underbrace{\mathbf{S}_{0,00} }_{ \text{rank}(\mathbf{\tilde F}) \times \text{rank}(\mathbf{\tilde F})} & \underbrace{\mathbf{S}_{0,01} }_{ \text{rank}(\mathbf{\tilde F}) \times \text{nullity}(\mathbf{\tilde F})} \\
        \underbrace{\mathbf{S}_{0,10} }_{ \text{nullity}(\mathbf{\tilde F}) \times \text{rank}(\mathbf{\tilde F})} & \underbrace{\mathbf{S}_{0,11} }_{ \text{nullity}(\mathbf{\tilde F}) \times \text{nullity}(\mathbf{\tilde F})}
    \end{bmatrix}.\]
    
\begin{proposition}\label{prop:finite_dependence_one_period}
    Suppose that $\mathbf{\tilde F}$ is the difference between two Markov transition matrices and $\mathbf{F}_0$ is a Markov transition matrix. If $\mathbf{S}_{0,01}  = 0$, then the model exhibit 1-period finite dependence, i.e., $\tilde{\mathbf F}^{(1)} = 0.$
\end{proposition}

Proposition \ref{prop:finite_dependence_one_period} elucidates the characteristics of Model 2 with respect to its 1-period finite dependence. This proposition delineates a sufficient condition for the establishment of 1-period finite dependence within the model framework. It posits that the null space of $\tilde{\mathbf{F}}$ remains invariant under the transformation induced by $\mathbf{F}_0$. Within this framework, a myriad of specific instances are encapsulated, including but not limited to, the phenomena of absorbing states. Additionally, the framework facilitates the identification of a pertinent weight vector, an assertion corroborated by the findings of Arcidiacono and Miller \cite{arcidiacono_miller_2019nonstationary}, who have contributed significantly to the discourse on nonstationary dynamic models.

\begin{proposition}\label{prop:finite_dependence_two_period}
    Suppose $\mathbf{\tilde F}$ is the difference between two Markov transition matrices and $\mathbf{F}_0$ is a Markov transition matrix. If $\mathbf{S}_{0,01} \mathbf{S}_{0,11} (\mathbf{I} - \mathbf{S}_{0,01}^+ \mathbf{S}_{0,01}) = 0$, then the model exhibit 1-period finite dependence, i.e., $\tilde{\mathbf F}^{(2)} = 0.$
\end{proposition}

The results presented herein offer insightful interpretations. Firstly, if a model exhibits 1-period finite dependence, denoted by $\mathbf{S}_{0,01} = 0$, it inherently satisfies the criteria for 2-period finite dependence. In other words, the condition for 2-period finite dependence is automatically met in such cases.
On the other hand, if $\mathbf{S}_{0,01}$ possesses full rank, the criterion for finite dependence is likewise fulfilled. It is typically observed that when the state transition is stochastically dependent on the actions, $\mathbf{S}_{0,01}$ is of full rank. This observation underscores a fundamental relationship between stochastic state transitions and the structural rank of $\mathbf{S}_{0,01}$, linking decision-making processes to the concept of finite dependence in dynamic models.

\textbf{[Add an example that $\mathbf{S}_{0,01}$ is of full rank.]}

\subsection{Computational details: overcoming the curse of dimensionality}\label{sec:computation_detail}
When dealing with sequential mapping, we can prevent the curse of dimensionality by avoiding the recursive multiplication of large matrices of $\mathbf{\tilde F}$ when solving for weight. The approach involves dividing the state into two parts: $\omega$ represents the states that are affected by the decision, while $z$ represents the exogenous state, which means that the state's transition is not influenced by the decision.
By the definition of $ \mathbf{\check w}_{t+1} $, we have the following formula and calculation result: 

\[
\begin{split}
  \mathbf{\check{w}}_{t+1}
  & = - \mathbf{\tilde{F}} \mathbf{F}_0 (\mathbf{\tilde{F}})^+ = \underbrace{-(\mathbf{\tilde{F}}_{\omega} \mathbf{F}_{\omega,0} \mathbf{\tilde{F}}_{\omega}^+ )}_{\mathbf{\check{w}}_{\omega, t+1}} \otimes \mathbf{F}_{z}, \\
  \mathbf{\tilde{F}} \mathbf{F}(\mathbf{w}) & = \mathbf{\tilde{F}}_{\omega} \mathbf{F}_{\omega,0} ( \mathbf{I} - (\mathbf{\tilde{F}}_{\omega}^+ \mathbf{\tilde{F}}_{\omega}) ) \otimes \mathbf{F}_{z},
\end{split}
\]
where $(\mathbf{\tilde F}^\top )^+$ is the Moore–Penrose inverse of matrix $\mathbf{\tilde F}^\top $.
The computation of $(\mathbf{\tilde F}^\top )^+$ is potentially costly, so we make use of this assumption. 
Partition the state variables into $(\omega,z)$ where the transition of $\omega$ depends on the action $d$ and the transition of $z$ does not. 

Then $\mathbf{F}_d = \mathbf{F}_{\omega, d} \otimes \mathbf{F}_{z}$. Then $\mathbf{\tilde F} = [\mathbf{F}_d - \mathbf{F}_0: d \in \mathcal{D}, d \neq 0] = [ ( \mathbf{F}_{\omega, d} - \mathbf{F}_{\omega, 0}) \otimes \mathbf{F}_{z}: d \in \mathcal{D}, d \neq 0] $. 
Let $ \mathbf{\tilde F}_{\omega} \equiv [ ( \mathbf{F}_{\omega, d} - \mathbf{F}_{\omega, 0})
: d \in \mathcal{D}, d \neq 0] $, then $(\mathbf{\tilde F}^\top )^+ = (\mathbf{\tilde F}_{\omega}^\top )^+ \otimes (\mathbf{F}_{z}^\top )^+ $.
Therefore, it is sufficient to only compute $(\mathbf{\tilde F}_{\omega}^\top )^+$. 

Plug in recursively to obtain $ \mathbf{\check w}_{t+2} $, 
\[\begin{split}
 \mathbf{\check w}_{t+2} & = - \big( \mathbf{\tilde F} \mathbf{F}(\mathbf{w}_{t+1}) \big) \mathbf{F}_0  (\mathbf{\tilde F})^+ = - \Big(\mathbf{\check w}_{t+1} \mathbf{\tilde F} + \mathbf{\tilde F} \mathbf{F}_0 \Big) \mathbf{F}_0  (\mathbf{\tilde F})^+ \\ 
 &= \underbrace{-  \Big(\big( \mathbf{\tilde F}_{\omega}  \mathbf{ F}_{\omega,0} +   \mathbf{\check w}_{\omega, t+1}  \mathbf{\tilde F}_{\omega} \big)  \mathbf{ F}_{\omega,0} \mathbf{\tilde F}_{\omega}^+  \Big)}_{ \mathbf{\check w}_{\omega, t+2} } \otimes (\mathbf{F}_z)^2.
 \end{split}\]
Plug in to get the bias correction mapping for $2$ period ahead value function:
\[
\begin{split}
\beta^3 \mathbf{\tilde F} \mathbf{F}(\mathbf{w}_{t+1}) \mathbf{F}(\mathbf{w}_{t+2}) & =  \beta^3 (\mathbf{\check w}_{t+2} \mathbf{\tilde F} + \mathbf{\check w}_{t+1} \mathbf{\tilde F}  \mathbf{F}_0 + \mathbf{\tilde F} (\mathbf{F}_0 )^2 ) \\ 
&= \beta^3 \Big( \mathbf{\check w}_{\omega, t+2} \mathbf{\tilde F}_{\omega} + \mathbf{\check w}_{\omega, t+1}  \mathbf{\tilde F}_{\omega} \mathbf{ F}_{\omega,0} + \mathbf{\tilde F}_{\omega} (\mathbf{ F}_{\omega,0})^2 \Big) \otimes (\mathbf{F}_z)^3
\end{split}
\]

More generally,
\[\mathbf{\tilde F} \mathbf{F}(\mathbf{w}_{t+1}), \ldots,  \mathbf{F}(\mathbf{w}_{t+\rho} ) =  \sum_{s=1}^{\rho} \mathbf{\check w}_{t+s} \mathbf{\tilde F} (\mathbf{F}_0)^{\rho-s} + \mathbf{\tilde F} (\mathbf{F}_0)^{\rho}. \]
Therefore, 
\[ \mathbf{\check w}_{t+\rho} = \underbrace{- \Big( \sum_{s=1}^{\rho-1} \mathbf{\check w}_{\omega, t+s} \mathbf{\tilde F}_{\omega} (\mathbf{F}_{\omega, 0})^{\rho-s} +  \mathbf{\tilde F}_{\omega} (\mathbf{F}_{\omega, 0})^{\rho}\Big) \mathbf{\tilde F}_{\omega}^+ }_{ \mathbf{\check w}_{\omega, t+\rho} }\otimes \mathbf{F}_z^{\rho}.\]
Then \[\mathbf{\tilde F} \mathbf{F}(\mathbf{w}_{t+1}), \ldots,  \mathbf{F}(\mathbf{w}_{t+\rho}) = \Big(\sum_{s=1}^{\rho} \mathbf{\check w}_{\omega, t+s}  \mathbf{\tilde F}_{\omega} (\mathbf{F}_{\omega,0})^{\rho-s} + \mathbf{\tilde F}_{\omega} (\mathbf{F}_{\omega,0})^{\rho} \Big) \otimes \mathbf{F}_Z^{\rho + 1}. \]

\section{Estimation}\label{sec:4_estimator}
We aim to estimate structural parameters, such as preferences, transition probabilities, and the discount factor $\beta$, for a dynamic discrete choice model with an infinite horizon. Agents, who make choices from a discrete set $\mathcal{D}$, seek to maximize their expected utility, which is the discounted sum of time-separable one-period utility functions $U_t(d_t,s_t)$. These decisions are based on state variables $s_t$ and take into account a Markovian transition probability $p_t(s_{t+1}|d_t, s_t)$.

The panel data for $N$ individuals comprises actions $d_t$ and a subset of state variables $x_t$. Here, $s_t$ includes both public information variables $x_t$, observable to researchers and agents, and private information variables $\epsilon_t$, known only to the agents. Our goal is to reliably estimate the structural parameter $\theta$ using this data, which spans $T_{\text{data}}$ periods.





\begin{assumption}[Type-I extreme value] \label{assumption:T1EV}
	The private information variables $\epsilon_t = \{ \epsilon_t(d): d \in \mathcal{D} \}$ is independently and identically distributed with a Type-I extreme value distribution, where $g_t(\epsilon_t|x_t) = \Pi_{d \in \mathcal{D}} \exp \{ - \epsilon_t(d) - \exp ( - \epsilon_t(d)) \}$. 
\end{assumption} 
  
\begin{example} [Hotz-Miller Inversion under Type I Extreme Value assumption]\label{example:hm_inversion}
	When  Assumption \ref{assumption:T1EV} holds,  the OCP mapping and its inverse mapping have a closed-form expression as follows:
	\[\begin{split}
		&\boldsymbol {\Lambda}_t(\tilde{\boldsymbol  v}_t(x))(d,x) = \frac{\exp(\tilde v_t(d,x))}{1+\sum_{d'=1}^D \exp(\tilde v_t(d',x))}\quad\text{for $d\in \mathcal{D}$}\quad\text{and}\quad \\ 
		& \boldsymbol {\Lambda}_t^{-1}(\boldsymbol  p_t(x)) = \log\left( (1 - \boldsymbol  p_t(x) \iota_{D}\t )^{-1}  \boldsymbol { p}_t(x)\right)
	\end{split}
	\]
	where  
	$\iota_D= (1,...,1)\t$ are  $D$-dimensional identity matrix and unit vector, respectively. Furthermore, $$\psi_d(\boldsymbol  p_t(x)) =  \gamma - \ln p_t(d,x)\quad\text{for $d\in \mathcal{D}$},$$  where $\gamma \approx 0.5772$ is the Euler constant.
\end{example}


\subsection{Conditional Choice Probabilities estimator under $\rho$-period finite dependence}

We  develop a computationally attractive estimator for non-stationary models by using the representation of value differences in equation (\ref{eq:rho-dependence}) together with the mapping from value differences to the conditional choice probabilities in (\ref{eq:ocp_mapping}).

The characterization of finite dependence discussed in Section \ref{sec:finite_dependence} can be numerically exploited to develop a computationally attractive estimator. 
Essentially, we don't need to solve for $\mathbf{w}_{x}$, because the conditional choice probability $\mathbf{p}(x_t, d_t) = \Lambda(\tilde{v}(x_t, d_t))$, so we only need to have a convenient expression for $\tilde{v}(x_t, d_t)$.

We assume that $x_t$ is written in a vector form  with $K$ variables as $x_t=(x_{1t},...,x_{Kt})\t$, where each element $x_t^k$ for $k=1,...,K$ has a finite support. For identification, we impose normalization condition that $u_t(x_t,0)=0$ for all $x_t \in\mathcal{X}$ and we further assume that
the instantaneous utility function is a linear function of $(x_{1t},...,x_{Kt})$:
\[
\tilde u_t(x_t,d) = u_t(x_t,d) = \theta_0^d + \theta_1^d x_{1t} + ... + \theta_K^d x_{Kt}:= x_t\t \theta^d
\]
with $\theta^d = (\theta_1^d,...,\theta_K^d)\t$ for $d=1,...,D$.

We estimate $\theta$ by minimizing the log-likelihood function.
\[\hat \theta^{ \rho-\text{FD} } = \arg\max \sum_{i=1}^N\sum_{t=1}^{T_{ {data}}} \log  \boldsymbol {\Lambda}(\tilde{\boldsymbol {v}}_t^{ \rho-\text{FD}}(x_t;\boldsymbol\theta))(d_t, x_t),
\]
where $\boldsymbol {\Lambda}(\cdot)$ is the OCP mapping and $\tilde{\boldsymbol {v}}^{ \rho-\text{FD}}(x_t)$ is a vector of \(\rho + 1\)-period-finite-dependent characterization of value difference: $\tilde{\boldsymbol {v}}^{ \rho-\text{FD}}(x_t) =\{ \tilde{ {v}}^{ \rho-\text{FD}}(x, d): d \in \mathcal{D} / \{ 0\} \}$.

Then the near-finite-dependent characterization minimizes the impact of omitting the value function two periods ahead.  
\begin{align} 
	\tilde{  v}_t^{ \rho-\text{FD}}(x_t, d_t)   =   \tilde u_t(x_t, d_t)  + \sum_{\tau=1}^{\rho}\beta^{\tau}  & \sum_{x_{t+\tau}\in\mathcal{X}} \tilde\kappa_{t+\tau-1}^{\mathbf{W}_{t+\tau-1}}(x_{t+\tau}|x_t, d_t) \nonumber \\ & \times  \left[ \bar u^{\mathbf{w}_{t+\tau}} _{t+\tau}(x_{t+\tau}|x_t, d_t) +\bar \psi^{\mathbf{w}_{t+\tau}}(\boldsymbol  p_{t+\tau}(x_{t+\tau})|x_t, d_t)\right] .
\end{align} 
Suppose that we estimate the conditional choice probabilities $\boldsymbol p_t$ and transition functions $f_t(x_{t+1}|x_t, d_t)$ in the first stage, the estimators of which are indicated by $\hat{\boldsymbol p}_t$ and $\hat f_t(x_{t+1}|x_t, d_t)$, respectively. Furthermore, suppose that given $\hat f_t(x_{t+1}|x_t, d_t)$, we have estimated the decision weights $\mathbf{W}_{t+\tau}$ under which a model exhibits $\rho$-period finite dependence, following a procedure discussed in the next section. Let $\hat{\mathbf{W}}_{t+\tau}(x_t,...,x_{t+\tau},d_t)=(\hat{\mathbf{w}}_{t}(x_t,x_{t+1},d_t),...,\hat{\mathbf{w}}_{t+\tau}(x_{t+\tau} | x_t, d_t))$ be the estimated weights. 

Define
\begin{align*}
& \bar  u^{\hat{\mathbf{w}}_{t+1} } _{t+1}(x_{t+1} | x_t, d_t):= \sum_{d' \in \mathcal{D}}\hat{ \mathrm{w}}_{t+1}(x_{t+1}, d'|x_t, d_t)    u_{t+1}(x_{t+1},d'),\\
&\bar\psi^{\hat{\mathbf{w}}_{t+1} }(\boldsymbol  p_{t+1}(x_{t+1}) | x_t, d_t):= \sum_{d' \in \mathcal{D}} \hat{\mathrm{w}}_{t+1}(x_{t+1}, d'|x_t, d_t) \psi_{d'}(\boldsymbol  p_{t+1}(x_{t+1})), \\
& \tilde{\kappa}_{t+\tau}^{\mathbf{W}_{t+\tau}}(x_{t+\tau+1} | x_t, d_t) := \left\{
\begin{array}{ll}
\tilde{f}_t(x_{t+1} | x_t, d_t), & \text{for } \tau = 0, \\
\sum_{x_{t+\tau} \in \mathcal{X}} \bar{f}_{t+\tau}^{\mathbf{w}_{t+\tau}}(x_{t+\tau+1} | x_{t+\tau}, x_t, d_t ) \\
\quad \times \tilde{\kappa}_{t+\tau-1}^{\mathbf{W}_{t+\tau-1}}(x_{t+\tau} | x_t, d_t), & \text{for } \tau = 1, \ldots, \mathcal{T},
\end{array}
\right. \\
&\text{with}\quad \bar f_{t+1}^{\hat{\mathbf{w}}_{t+1}}(x_{t+2}|x_{t+1}, x_t, d_t):=\sum_{d' \in \mathcal{D}} \hat{\mathrm{w}}_{t+1}(x_{t+1}, d'|x_t, d_t) f_{t+1}(x_{t+2}|x_{t+1},d').
\end{align*}

In this subsection, for brevity, we assume that both the transition function and the decision weights $\mathbf{W}_{t+\rho}$ under which a model exhibits $\rho$-period-finite dependence are known to econometricians. Furthermore, for identification, we assume that $u_t(x_t,0)=0$ for all $x_t\in\mathcal{X}$ and the utility function is specified in linear in parameters as:
\begin{equation}
     \tilde u_t(x_t,d) = u_t(x_t,d) = x_t\t \bs\theta^{d}\quad\text{for $d\in\{1,2,...,D\}$.}
\end{equation}
 \begin{example}[Type-1 Extreme Value continued]

For T1EV, $\psi_d(\boldsymbol  p_t(z)) =  \gamma - \ln p_t(d,z)$ with $\gamma \approx 0.5772$. Suppose that we have solved the weight $\mathbf{W}_{t+\rho}$ under which a model exhibits the dependence of the $\rho + 1$ period. 
Then the finite dependence estimator estimator is constructed as 
\[\hat \theta^{ \rho-\text{FD} } = \arg\max \sum_{i=1}^N\sum_{t=1}^{T_{ {data}}} \log  \boldsymbol{\Lambda}( \mathbf{H}^{\mathbf{W}_{t+\rho}}(d_t, x_t)^\top \boldsymbol{\theta}_{d_t} + \mathbf{h}^{\mathbf{W}_{t+\rho}}(d_t, x_t) ),
\]
where \begin{align*}
    \mathbf{H}^{\mathbf{W}_{t+\rho}}(d_t, x_t) & = x_t + \sum_{\tau=1}^{\rho}\beta^{\tau} \sum_{x_{t+\tau}\in\mathcal{X}} \tilde\kappa_{t+\tau-1}^{\mathbf{W}_{t+\tau-1}}(x_{t+\tau}|x_t, d_t) \bar {\mathbf{x}}_{t+\tau}(x_t) \\
    \mathbf{h}^{\mathbf{W}_{t+\rho}}(d_t, x_t) & = \sum_{\tau=1}^{\rho}\beta^{\tau} \sum_{x_{t+\tau}\in\mathcal{X}} \tilde\kappa_{t+\tau-1}^{\mathbf{W}_{t+\tau-1}}(x_{t+\tau}|x_t, d_t) \bar \psi^{\mathbf{w}_{t+\tau}}(\boldsymbol  p_{t+\tau}(x_{t+\tau})|x_t, d_t) .\\
\end{align*}

\end{example}

\section{Monte Carlo simulation}\label{sec:5_MCsimulation}

In our simulation, we analyze an entry/exit problem within a dynamic framework that can be adapted to both stationary and non-stationary environments, with a focus on non-stationary transition dynamics. Our model extends \cite{Aguirregabiria2016} by incorporating past entry decisions into profitability, represented by the structural vector $\theta = (\theta_0^{ {VP}}, \theta_1^{ {VP}}, \theta_2^{ {VP}},\theta_0^{ {FC}},\theta_1^{ {FC}},\theta_0^{ {EC}},\theta_1^{ {EC}})$. This approach is further refined with lagged productivity effects as in \cite{kasahara2018estimation}, which challenges the finite dependence structure of the model.
In a Monte Carlo simulation designed to compare the effectiveness of estimation methods in various environments, the FD and FD2 methods, which are alternatives to the traditional Hotz-Miller estimator, demonstrated substantially quicker computation times, particularly in scenarios with large state spaces. This advantage is crucial in stationary environments. For models lacking 1-period finite dependence, such as those with absorbing states discussed in \cite{arcidiacono2011conditional}, FD2 remains applicable due to its reliance on 2-period finite dependence. In non-stationary environments, where the Hotz-Miller (HM) estimator is unsuitable, the FD class of estimators becomes essential. The FD2 estimator excels in these settings, exhibiting minimal bias. In contrast, both the FD method and FD without the $x_t$-dependent weight characterization are less effective, as they tend to produce biased estimates.

The state variables are $x = (z_1,z_2,z_3,z_4,\omega,y)$, where $y$ signifies a firm's market presence, and $\epsilon$ captures unobserved heterogeneity. Firms, observing $(x,\epsilon)$, decide whether to operate ($d = 1$) or not ($d = 0$).

The state at time $t$ is $(z_{1,t},\ldots,z_{4,t},\omega_t,y_t,\epsilon_t)$, with $y_t$ determined by the previous period's action $a_{t-1}$. The firm's flow payoff, as described in equation (\ref{eq:stationary_flow_payoff}), comprises variable profits ($VP$), fixed ($FC$), and entry costs ($EC$). Operating firms pay $FC$ and, if entering, additional $EC$.
\begin{equation}
	\label{eq:stationary_flow_payoff}
	\begin{split}
		u( d_t ,x_t;\theta) & =  d_t(  { {VP}}_t -  {EC}_t -  {FC}_t) \\
		\text{ where }  { {VP}}_t & =  \exp(\omega )[\theta_0^{ { {VP}}} + \theta_1^{ { {VP}}} z_{1t} + \theta_2^{ { {VP}}} z_{2t}] \\
		 {FC}_t & =  [\theta_0^{ {FC}} + \theta_1^{ {FC}}z_{3t}] \\
		 {EC}_t & = (1 - y_t) [\theta_0^{ {EC}} + \theta_1^{ {EC}}z_{4t}].
	\end{split}
\end{equation}

\begin{table}[!htbp]
	\caption{Parameters in Data Generating Process}
	\centering
    \small
	\begin{tabular}{l|rrr}
		\toprule
		Payoff Parameters:  & $\theta_0^{ { {VP}}}  = 0.5$ & $ \theta_1^{ { {VP}}} = 1.0 $ &  $\theta_2^{ { {VP}}}  = - 1.0 $ \\
		~ &  $\theta_0^{ {FC}}  = 0.5$ & $ \theta_1^{ {FC}} = 1.0 $ & ~  \\
		~ &  $\theta_0^{ {EC}}  = 1.0$ & $ \theta_1^{ {EC}} = 1.0 $ & ~  \\
		\midrule
		Each $z_k$ state variable & $z_{kt}$ is AR(1), & $\gamma_0^{k} = 0,$ & $\gamma_1^{k} = 0.9$, $\sigma$ = 1\\
		Productivity & $\omega_t$ is AR(1), & $\gamma_0^{\omega} = 0,$ & $\gamma_1^{\omega} = 0.9$, $\sigma$ = 1\\
		Discount Factor & $\beta = 0.95$.\\
		\bottomrule
	\end{tabular}
	\label{tab:DGP_parameters}
\end{table}

The exogenous shocks $(z_1,z_2,z_3,z_4,\omega)$ independently follow $AR(1)$ processes. We apply the method by \cite{tauchen1986finite} to discretize these processes into a finite state space. Let $K_z$ and $K_o$ represent the number of discrete points for $z_{j}$ and $\omega$, respectively. The resulting state space has a dimensionality of $X = 2 * K_z^4 * K_o$.
We let $z_1, \ldots, z_4$'s transition density be independent of the action chosen.
To discretize each $z_j$, we define $K_z$ support points with interval widths $w_j^{(k)} = z_j^{(k+1)} - z_j^{(k)}$. The latent variable $\tilde{z}_{jt}$, following the process $\tilde{z}_{jt} = \gamma_0^j + \gamma_1^j \tilde{z}_{j,t-1} + e_{jt}$ with $e_{jt} \overset{i.i.d}{\sim} \mathcal{N}(0,\sigma_j^2)$, is discretized by determining its normal distribution with mean $\mu_k = \frac{\gamma_0^k}{1 - \gamma_1^k}$ and variance $\sigma_k^2 = \frac{\sigma^k}{1 - (\gamma_1^k)^2}$. The discretization grid is established using quantiles at $0.5/K_z$ and $1 - 0.5/K_z$ of this distribution.

We model the productivity shock $\tilde{\omega}_{t}$ as a function of past actions, following the process $\tilde{\omega}_{t} = \gamma_{0t}^{\omega} + d_{t-1} \gamma_a + \gamma_1^{\omega} \tilde{\omega}_{t-1} + e_{jt}$, where the disturbance term $e_{jt}$ is independently and identically distributed as $\mathcal{N}(0,\sigma_j^2)$. In the stationary context, $\gamma_{0t}^{\omega}$ is set to zero for all $t$. Conversely, in the non-stationary scenario, $\gamma_{0t}^{\omega}$ varies with time, assuming values $[-0.8, 0.8, 0, -0.3]$ for $t = 1, 2, 3, 4$, respectively. Because of the nature of the transition matrix, I fix the maximum to be 1 and the minimum of the grid point to be -1.

\subsection{Monte Carlo simulation weight solving time}

Suppose we are interested in the broadening of the weight concept to encompass multiple periods. In this section, we engage in a comparative analysis of the time required for weight computation utilizing the proposed algorithm against conventional time measures. Specifically, the "Time" columns in the table delineate the duration necessary for computing the weight, as well as for resolving the functions $\mathbf{H}$ and $\mathbf{h}$. To establish a benchmark for comparison, we also document the time expended in executing the Hotz-Miller (HM) inversion process. It is noteworthy that the computation of the weight and the subsequent linear components of the estimators is relatively cost-efficient compared to an HM inversion, thereby underscoring the computational attractiveness of our proposed method.

Furthermore, in the "Norm" columns, we report the largest singular values of the perturbation $\tilde{\mathbf F}^{(\rho)}$ for $\rho = 1,2$. This metric can be construed as an indication of the influence that the $\rho+1$ value function exerts on the Conditional Choice Probability (CCP) estimator. The norm, when calculated with $\rho=1$, is not ideally minimal; however, for $\rho = 2$, it approaches zero, suggesting a negligible impact when the full solution is contemplated. 

When the weight is resolved sequentially, first by calculating $\mathbf{w}_{t+1}$ and then $\mathbf{w}_{t+2}$ the norm is reduced compared to a single period analysis but does not achieve a value sufficiently close to zero. This finding is instrumental in illustrating the trade-offs between computational efficiency and accuracy within the context of extending the weight to multiple periods. 

\begin{table}[htbp]
    \centering
    \scriptsize
    \caption{Timing Analysis and Weight Vector Computation for FD Algorithm}
    \label{tab:sample-data}
    \begin{threeparttable}    
    \begin{tabular}{l @{\hspace{1.5cm}} rrrrr @{\hspace{1.5cm}} rr}
        \toprule
        N state &  \multicolumn{5}{c}{Time}      &  \multicolumn{2}{c}{Norm} \\
        & Solving & HM Inverse & FD 1      & compute $\mathbf{w}$ & FD 2  & $\tilde{\mathbf F}^{(1)}$ & $\tilde{\mathbf F}^{(2)}$ \\
        \cmidrule(r){2-6} 
        \cmidrule(r){7-8} 
        \multicolumn{8}{c}{Finite Dependence $\gamma = 0$} \\ \midrule
        
        64& 0.022  & 0.001& 0.001  & 0.001    & 0.001       & 3.27E-17& 3.78E-17         \\
        96& 0.027  & 0.000& 0.001  & 0.001    & 0.002       & 1.94E-17& 1.71E-17         \\
        128& 0.047  & 0.001& 0.000  & 0.000    & 0.002       & 5.83E-17& 6.69E-17         \\
        160& 0.046  & 0.001& 0.001  & 0.001    & 0.002       & 6.39E-17& 4.61E-17         \\
        324& 0.118  & 0.006& 0.002  & 0.001    & 0.004       & 3.37E-17& 3.93E-17         \\
        486& 0.923  & 0.022& 0.006  & 0.002    & 0.027       & 1.99E-17& 1.78E-17         \\
        648& 2.011  & 0.064& 0.013  & 0.003    & 0.051       & 6.00E-17& 6.95E-17         \\
        810& 2.328  & 0.073& 0.020  & 0.007    & 0.072       & 6.58E-17& 4.79E-17         \\
        1024 & 3.948  & 0.122& 0.029  & 0.008    & 0.114       & 3.44E-17& 4.04E-17         \\
        1536 & 9.555  & 0.401& 0.064  & 0.014    & 0.265       & 2.03E-17& 1.83E-17         \\
        2048 & 16.858 & 0.859& 0.111  & 0.028    & 0.445       & 6.13E-17& 7.15E-17         \\
        2560 & 25.840 & 1.451& 0.184  & 0.044    & 0.683       & 6.72E-17& 4.93E-17         \\
        2500 & 24.506 & 1.285& 0.165  & 0.046    & 0.645       & 3.48E-17& 4.11E-17         \\
        3750 & 56.645 & 3.810& 0.382  & 0.103    & 1.445       & 2.06E-17& 1.86E-17         \\
        5000 & 96.439 & 8.255& 0.717  & 0.183    & 2.528       & 6.19E-17& 7.27E-17         \\
        6250 & 138.782& 15.692& 1.088  & 0.311    & 3.938       & 6.80E-17& 5.01E-17         \\
        15552& 883.166& 177.803 & 6.113  & 1.710    & 24.993      & 8.65E-17& 8.61E-17         \\
        18144& 1270.952& 260.719 & 8.800  & 2.503    & 33.075      & 7.25E-17& 7.63E-17      \\
        \midrule
        \multicolumn{8}{c}{Non-finite Dependence $\gamma = 0.5$} \\
        \midrule
        
        64   & 0.022     & 0.000   & 0.000                 & 0.000       & 0.001                      & 1.45E-01               & 5.03E-17                        \\
        96   & 0.031     & 0.000   & 0.001                 & 0.001       & 0.001                      & 1.77E-01               & 7.98E-17                        \\
        128  & 0.044     & 0.001   & 0.000                 & 0.000       & 0.001                      & 1.95E-01               & 4.99E-16                        \\
        160  & 0.052     & 0.002   & 0.000                 & 0.000       & 0.002                      & 1.99E-01               & 6.04E-14                        \\
        324  & 0.094     & 0.005   & 0.002                 & 0.001       & 0.004                      & 1.49E-01               & 5.22E-17                        \\
        486  & 0.862     & 0.025   & 0.008                 & 0.003       & 0.026                      & 1.82E-01               & 8.30E-17                        \\
        648  & 1.532     & 0.042   & 0.012                 & 0.004       & 0.049                      & 2.01E-01               & 5.19E-16                        \\
        810  & 2.392     & 0.066   & 0.023                 & 0.006       & 0.088                      & 2.05E-01               & 6.28E-14                        \\
        1024 & 4.128     & 0.139   & 0.032                 & 0.008       & 0.119                      & 1.53E-01               & 5.37E-17                        \\
        1536 & 10.325    & 0.414   & 0.070                 & 0.020       & 0.232                      & 1.86E-01               & 8.53E-17                        \\
        2048 & 16.992    & 0.856   & 0.109                 & 0.026       & 0.421                      & 2.05E-01               & 5.34E-16                        \\
        2560 & 27.044    & 1.475   & 0.170                 & 0.045       & 0.657                      & 2.09E-01               & 6.46E-14                        \\
        2500 & 24.567    & 1.275   & 0.164                 & 0.046       & 0.655                      & 1.54E-01               & 5.46E-17                        \\
        3750 & 55.723    & 3.795   & 0.378                 & 0.097       & 1.456                      & 1.88E-01               & 8.67E-17                        \\
        5000 & 96.322    & 8.304   & 0.751                 & 0.188       & 2.559                      & 2.07E-01               & 5.42E-16                        \\
        6250 & 127.850   & 15.018  & 1.065                 & 0.282       & 3.874                      & 2.11E-01               & 6.56E-14                        \\
        15552& 928.958   & 167.627 & 6.516                 & 1.774       & 24.006                     & 2.23E-01               & 3.84E-11                        \\
        18144& 1297.067  & 252.541 & 8.561                 & 2.435       & 32.995                     & 2.32E-01               & 4.61E-08       \\                
        \bottomrule
    \end{tabular}
\begin{tablenotes}
    \item[1] Computation Time: Details the duration required for different computational tasks.\\
        \item[1.a] Solving Time: Duration for a single iteration of solving for the model's fixed point for a candidate parameter.\\
        \item[1.b] Hotz-Miller Inversion Time: Time required for inverting the matrix \(I - \beta \mathbf{F}^{\mathbf{P}}\), a critical step for the \cite{Hotz1993} estimation procedure.\\
        \item[1.c] Computation for \(\mathbf{w}\): Time taken to compute the weight vectors as described in Section \ref{sec:computation_detail}.\\
        \item[1.d] FD1/FD2 Time: Time required to execute the FD1/FD2 estimator after obtaining the weight vector.\\
    \item[2] Norms of \(\tilde{\mathbf F}^{(1)}\) and \(\tilde{\mathbf F}^{(2)}\): These are the largest singular values from the bias correction mappings. They indicate the influence of future value function on the current one. The near-zero value of \(\tilde{\mathbf F}^{(2)}\) under all conditions suggests that it is possible to identify weight vectors that ensure models exhibit the finite dependence property.
\end{tablenotes}
\end{threeparttable}
    
\end{table}

\begin{figure}[htbp]
    \centering
    \caption{Comparative Computation Times for FD Weight Solving and HM Inversion as a Function of State Space Size}
    \includegraphics[width=0.6\linewidth]{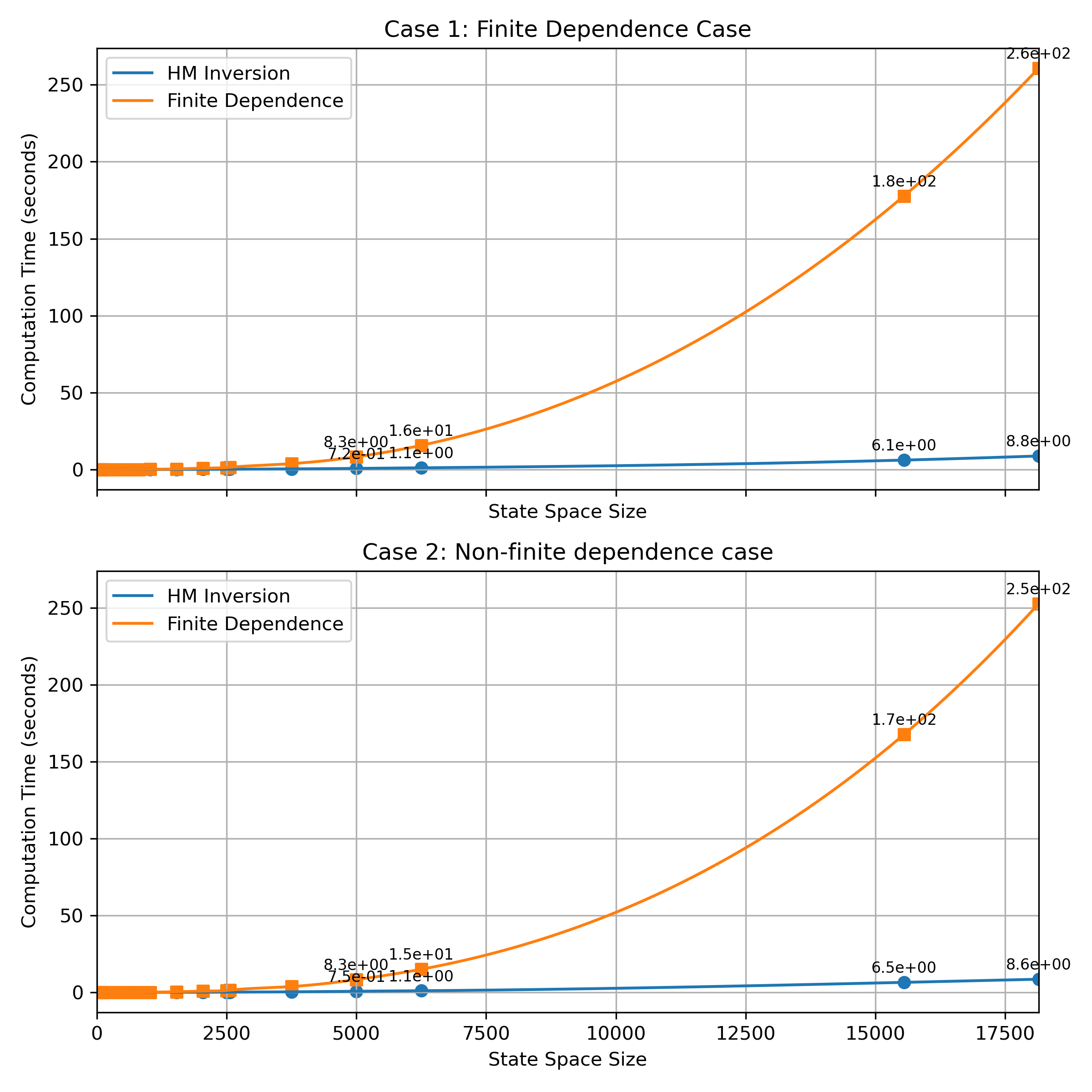}
    \label{fig:computation_time_state_space}
\begin{tablenotes}
    \footnotesize
    Note: This figure contrasts the computation time required for Fixed-Decision (FD) weight solving with that of Hotz-Miller (HM) inversion, relative to the size of the state space. The cubic increase in computation time for HM inversion with state space size suggests a higher computational complexity. Conversely, FD weight solving demonstrates a sub-cubic rate of increase, implying enhanced scalability for larger state spaces. Such insights are pivotal for evaluating the computational trade-offs between initial setup times and long-term scalability when applying the FD and HM methods to extensive models.
\end{tablenotes}
\end{figure}
\subsection{Stationary 2-step estimator}

\begin{table}[htbp]
	\centering
	\scriptsize
	\caption{Mean Estimates from Two-Step Estimation: Sample with $N = 30$, $T = 40$}
	\begin{threeparttable}
		\begin{tabular}{lccccccc cc}
			\toprule
			Estimator & $\theta^\mathrm{\mathrm{\mathrm{VP}}0}_m$  & $\theta^\mathrm{\mathrm{\mathrm{VP}}1}_m$ & $\theta^\mathrm{\mathrm{\mathrm{VP}}2}_m$& $\theta^\mathrm{\mathrm{FC}0}_m$& $\theta^\mathrm{\mathrm{FC}1}_m$& $\theta^\mathrm{\mathrm{EC}0}_m$  & $\theta^\mathrm{\mathrm{EC}1}_m$ & Time      & Time Inv    \\
			\midrule
         &    \multicolumn{9}{c}{$X = 2500$, $\gamma_a = 0$}       \\
         \cmidrule{2-10} 
HM        & 0.51   & 1.021  & -1.018 & 0.489  & 1.012  & 0.988  & 1.013  & 1.2115 & 1.1413 \\
          & (0.149) & (0.091) & (0.091) & (0.237) & (0.097) & (0.223) & (0.137) &        &        \\
FD   & 0.508  & 1.017  & -1.013 & 0.502  & 1.007  & 0.968  & 1.004  & 0.2002 & 0.1494 \\
          & (0.151) & (0.079) & (0.077) & (0.275) & (0.09)  & (0.238) & (0.094) &        &        \\
FD2  & 0.508  & 1.012  & -1.009 & 0.515  & 1.006  & 0.961  & 1.002  & 0.6375 & 0.0918 \\
          & (0.17)  & (0.046) & (0.047) & (0.306) & (0.06)  & (0.242) & (0.072) &        &        \\
         \cmidrule{2-10} 
 & \multicolumn{9}{c}{$X = 2500$, $\gamma_a = 0.5$}                 \\ \cmidrule{2-10} 
HM        & 0.514  & 1.019  & -1.017 & 0.518  & 1.015  & 0.961  & 0.997  & 1.1918 & 1.1178 \\
          & (0.149) & (0.09)  & (0.094) & (0.319) & (0.099) & (0.282) & (0.143) &        &        \\
FD  & 0.506  & 1.044  & -1.041 & 0.386  & 1.014  & 0.963  & 1.006  & 0.0957 & 0.0454 \\
          & (0.146) & (0.097) & (0.1)   & (0.391) & (0.097) & (0.289) & (0.107) &        &        \\
FD2 & 0.52   & 1.03   & -1.025 & 0.527  & 1.019  & 0.956  & 1.001  & 0.6175 & 0.0903 \\
          & (0.149) & (0.122) & (0.126) & (0.395) & (0.125) & (0.288) & (0.12)  &        &        \\
         \cmidrule{2-10} 
 & \multicolumn{9}{c}{$X = 12960$, $\gamma_a = 0$}        \\ \cmidrule{2-10} 
HM   & 0.509  & 1.018  & -1.02  & 0.515  & 1.012  & 1.013  & 1.015  & 94.4154 & 94.3492 \\
     & (0.167) & (0.086) & (0.082) & (0.286) & (0.081) & (0.225) & (0.118) &         &         \\
FD  & 0.509  & 1.018  & -1.021 & 0.515  & 1.01   & 1.004  & 1.011  & 1.2910  & 1.2427  \\
     & (0.168) & (0.082) & (0.078) & (0.303) & (0.081) & (0.216) & (0.099) &         &         \\
FD2 & 0.511  & 1.023  & -1.031 & 0.514  & 1.016  & 1.008  & 1.016  & 15.8718 & 2.4842  \\
     & (0.156) & (0.092) & (0.093) & (0.234) & (0.081) & (0.22)  & (0.098) &         &         \\
         \cmidrule{2-10} 
 & \multicolumn{9}{c}{$X = 12960$, $\gamma_a = 0.5$}  \\ \cmidrule{2-10} 
HM   & 0.502  & 1.021  & -1.016 & 0.474  & 1.006  & 1.018  & 1.02   & 94.3912 & 94.3184 \\
     & (0.164) & (0.093) & (0.091) & (0.364) & (0.093) & (0.287) & (0.137) &         &         \\
FD  & 0.484  & 1.022  & -1.018 & 0.31   & 1.01   & 1.009  & 1.027  & 1.3737  & 1.3106  \\
     & (0.209) & (0.046) & (0.043) & (0.965) & (0.093) & (0.274) & (0.109) &         &         \\
FD2 & 0.501  & 1.028  & -1.019 & 0.446  & 1.005  & 1.005  & 1.013  & 15.7964 & 2.4861  \\
     & (0.169) & (0.115) & (0.113) & (0.466) & (0.12)  & (0.275) & (0.112) &         &         \\
			\bottomrule
		\end{tabular}
\begin{tablenotes}
  \item[1] Data generation parameter set $\theta$: $\theta^\mathrm{VP}_0=0.5$, $\theta^\mathrm{VP}_1=1.0$, $\theta^\mathrm{VP}_2=-1.0$, $\theta^\mathrm{FC}_0=0.5$, $\theta^\mathrm{FC}_1=1.0$, $\theta^\mathrm{EC}_0=1.0$, $\theta^\mathrm{EC}_1=1.0$.\\
  \item[2] Table entries represent the mean of 100 Monte Carlo simulation estimates (first row) and the corresponding standardized mean squared errors (second row) for each estimator. In the 1-period finite dependence scenario ($\gamma_a = 0$), all estimators are consistent. In contrast, for non-1-period finite dependence ($\gamma_a = 0.5$), only the HM and FD2 estimators maintain consistency.\\
  \item[3] "Time" indicates the computation time required for each estimator. "Time Inv" refers to the time spent on matrix inversion for the HM estimator. For FD estimators, "Time Inv" represents the time taken to compute the optimal weights. The FD and FD2 estimators outperform the HM in terms of speed, particularly when dealing with larger state spaces.
\end{tablenotes}	
\end{threeparttable}\label{tab:two_step_large}
\end{table}

The table presents the performance of two-step estimators under a stationary model with different configurations, specifically focusing on the mean estimates for various parameters across different methods. The simulations are based on a sample size of $N = 30$ and $T = 40$, with two distinct values for the exogenous variable $X$ and the autocorrelation parameter $\gamma_a$. The estimated parameters are denoted by $\theta^\mathrm{VP}_m$ for variance parameters, $\theta^\mathrm{FC}_m$ for first coefficient parameters, and $\theta^\mathrm{EC}_m$ for the second coefficient parameters. The computational time and time involved in inversion (if applicable) are also reported.

Under the scenario where $X = 2500$ and $\gamma_a = 0$, it is observed that the estimators HM, FD, and FD2 provide mean estimates close to the true parameters values of $\theta_0^{\mathrm{VP}}=0.5$, $\theta_1^{\mathrm{VP}}=1$, and $\theta_2^{\mathrm{VP}}=-1$. The standard errors are relatively small, indicating precision in the estimates. The time taken for computations is significantly less for the FD and FD2 estimators compared to the HM estimator, with FD being the fastest. When $\gamma_a$ is increased to $0.5$, the mean estimates of the parameters remain relatively stable, with some increase in standard errors, suggesting that the presence of autocorrelation affects the precision of the estimates. In this case, FD remains the fastest, suggesting efficiency in computational time.

When the value of $X$ is increased to $12960$ while keeping $\gamma_a = 0$, the mean estimates remain consistent with the true parameter values, and the standard errors do not show a significant increase, which indicates that the increase in the exogenous variable does not substantially affect the estimator's performance. However, the computational time for the HM estimator increases drastically, highlighting its computational inefficiency for larger datasets. In the FD and FD2 estimators, the time increases as well, but not as dramatically, showing their relative computational efficiency.

With both $X = 12960$ and $\gamma_a = 0.5$, the mean estimates are still consistent with the true parameters, although the standard errors are somewhat larger compared to the case with no autocorrelation, which is expected as autocorrelation generally makes estimations more challenging. The computational time remains high for the HM estimator, while FD and FD2 show better computational efficiency. 

Overall, the results indicate that the FD and FD2 are more computationally efficient than the HM estimator, especially as the size of the exogenous variable increases. The estimators' performance is relatively robust to the presence of autocorrelation, although it does increase the standard errors slightly. The results highlight the importance of considering both estimation accuracy and computational efficiency when choosing an estimator for practical applications.

\subsection{Non-stationary 2-step estimator}

The table reports the estimation of parameters for non-stationary models. Traditional methods like the Rust successive approximation or Hotz and Miller's approach fall short, mainly due to the challenge of determining the future value \( V_{Tdata} \) at the terminal period. 

The innovative FD2 method estimates the weighted transition probabilities with 2 period weights \(\mathbf{w}_{t+1}(x_{t+1} | x_t)\) and \(\mathbf{w}_{t+2}(x_{t+2} | x_t)\), and it performs remarkably well in minimizing bias and mean squared error (MSE), as seen from the values in brackets. 
Overall, the FD2 method outperforms all other methods, with  smaller bias and lower MSE.

It is worth noting that if we allow the weight to be specific to the value of $x_t$, as suggested by the weight search algorithm proposed by Arcidiacono and Miller (2019), we can keep the advantage of the Kronecker product structure of the transition densities. We refer to the weight-solving proposed by Arcidiacono and Miller (2019) as FD without $x_t$-dependent weights. However, using this method can significantly increase the time required for solving weights in high-dimensional models.
Furthermore, if we do not allow the weight to be based on $x_t$, our approach delivers inferior results compared to the FD method in terms of attaining a lower \(\rho\) for \(\mathbf{w}_{t+1}\). This implies that the FD method is more effective in discovering the optimal weights and providing a more precise estimation of the model parameters.

In the table, we can observe the effectiveness of the FD2 method in estimating the non-stationary model parameters compared to the traditional FD method and the FD method without $x_t$-dependent weights. The table presents the estimated values for the parameters VP0, VP1, VP2, FC0, FC1, EC0, and EC1, as well as the time and values for $\rho_1$ and $\rho_2$. The numbers in brackets denote the respective bias and mean squared error (MSE).

When analyzing the table, it is evident that the FD2 method has a smaller bias and lower MSE compared to both the FD method and the FD method without $x_t$-dependent weights. This highlights the superiority of the FD2 method in estimating nonstationary model parameters more accurately. Additionally, the time required to solve weights in high-dimensional models can be significantly reduced by using the weight search algorithm proposed by Arcidiacono and Miller (2019).
In conclusion, the table demonstrates the performance of the FD2 method in estimating the parameters of nonstationary models, providing more accurate and precise results than the FD and FD without $x_t$-dependent weights methods.

The FD method without the $x_t$ dependence requires a brute-force approach to determine the weights, resulting in a time-intensive process. However, incorporating $x_t$-dependent weights allows for the utilization of Kronecker products, which streamlines the computation of weights and significantly enhances efficiency. This refinement in weight characterization results in a more practical and expedient estimation process for complex models.

\begin{table}[htbp]
\centering
\caption{Estimator Performance in a Non-Stationary Environment with Time Horizon \(T=4\)}
\begin{threeparttable}
\scriptsize
\begin{tabular}{lcccccccccc}
\toprule
{} & VP0 & VP1 & VP2 & FC0 & FC1 & EC0 & EC1 & Time (s) & \(\tilde{\mathbf F}^{(2)}_1\) & \(\tilde{\mathbf F}^{(2)}_2\) \\
\midrule
True $\theta$ & 0.5 & 1.0 & -1.0 & 0.5 & 1.0 & 1.0 & 1.0\\
\midrule
\textbf{} & \multicolumn{10}{c}{\textbf{X=1024, $\gamma_a$ = 0}} \\
\midrule
\textbf{FD} & 0.498 & 1.010 & -1.008 & 0.526 & 1.022 & 0.985 & 1.008 & 0.12 & 5.18E-17 & 2.63E-17 \\
\textbf{} & (0.120)& (0.094)& (0.091)& (0.194)& (0.071)& (0.219)& (0.090)& 0.07 & & \\
\textbf{FD2} & 0.502 & 1.004 & -1.002 & 0.511 & 1.000 & 0.990 & 1.000 & 0.37 & 9.38E-15 & 2.62E-17 \\
\textbf{} & (0.056)& (0.027)& (0.022)& (0.137)& (0.001)& (0.215)& (0.001)& 0.15 & & \\
\textbf{FD} & 0.495 & 1.016 & -1.015 & 0.491 & 1.002 & 1.002 & 0.996 & 1046.91 & 2.37E-16 & 4.31E-16 \\
\textbf{(no $x_t$-dependent)} & (0.118)& (0.096)& (0.093)& (0.190)& (0.067)& (0.216)& (0.090)& 1046.86 & & \\ \midrule
\textbf{} & \multicolumn{10}{c}{\textbf{X=1024, $\gamma_a$ = 0.5}} \\ \midrule
\textbf{FD} & 0.477 & 1.103 & -1.104 & 0.111 & 0.984 & 1.023 & 1.003 & 0.11 & 1.01E-01 & 1.58E-01 \\
\textbf{} & (0.132)& (0.135)& (0.134)& (0.445)& (0.081)& (0.227)& (0.086)& 0.06 & & \\
\textbf{FD2} & 0.497 & 1.011 & -1.011 & 0.467 & 1.010 & 1.029 & 1.002 & 0.21 & 3.56E-17 & 6.83E-17 \\
\textbf{} & (0.115)& (0.091)& (0.092)& (0.254)& (0.099)& (0.234)& (0.099)& 0.16 & & \\
\textbf{FD} & 0.432 & 1.117 & -1.118 & 0.067 & 0.953 & 0.905 & 1.028 & 767.69 & 1.13E-01 & 3.11E-01 \\
\textbf{(no $x_t$-dependent)} & (0.143)& (0.156)& (0.153)& (0.468)& (0.094)& (0.216)& (0.125)& 767.64 & & \\

\midrule
\textbf{} & \multicolumn{10}{c}{\textbf{X=2500, $\gamma_a$ = 0}} \\ \midrule
\textbf{FD} & 0.505 & 1.017 & -1.017 & 0.515 & 1.027 & 0.976 & 0.998 & 0.44 & 5.24E-17 & 2.66E-17 \\
\textbf{} & (0.111)& (0.083)& (0.082)& (0.172)& (0.088)& (0.249)& (0.095)& 0.40 & & \\
\textbf{FD2} & 0.498 & 1.000 & -1.000 & 0.502 & 1.000 & 0.991 & 1.000 & 1.11 & 9.53E-15 & 2.66E-17 \\
\textbf{} & (0.047)& (0.000)& (0.000)& (0.140)& (0.001)& (0.245)& (0.001)& 0.91 & & \\
\textbf{FD} & 0.502 & 1.021 & -1.021 & 0.479 & 1.007 & 0.997 & 0.986 & 23634.46& 1.36E-16 & 1.87E-16 \\
\textbf{(no $x_t$-dependent)} & (0.111)& (0.084)& (0.082)& (0.169)& (0.084)& (0.247)& (0.095)& 23634.41& & \\ \midrule
\textbf{} & \multicolumn{10}{c}{\textbf{X=2500, $\gamma_a$ = 0.5}} \\ \midrule
\textbf{FD} & 0.476 & 1.115 & -1.114 & 0.103 & 1.000 & 1.008 & 1.000 & 0.41 & 1.03E-01 & 1.60E-01 \\
\textbf{} & (0.133)& (0.155)& (0.154)& (0.450)& (0.099)& (0.258)& (0.096)& 0.36 & & \\
\textbf{FD2} & 0.489 & 1.009 & -1.005 & 0.485 & 1.021 & 1.005 & 0.999 & 0.95 & 3.62E-17 & 6.94E-17 \\
\textbf{} & (0.119)& (0.092)& (0.082)& (0.232)& (0.125)& (0.270)& (0.110)& 0.91 & & \\
\textbf{FD} & 0.424 & 1.124 & -1.123 & -0.026 & 0.952 & 1.059 & 1.002 & 6725.81 & 1.66E-01 & 2.33E-01 \\
\textbf{(no $x_t$-dependent)} & (0.155)& (0.172)& (0.170)& (0.561)& (0.108)& (0.261)& (0.098)& 6725.76 & & \\
\bottomrule
\end{tabular}
\begin{tablenotes}
  \item \textbf{Note:} This table compares estimator performance in non-stationary conditions over a time horizon of $T=4$. True parameter values ($\theta$) are included for reference. \\
  \item[1] The evaluated estimators, FD and FD2, are analyzed both with and without weights depends on $x_t$.\\
  \item[2] The coefficients VP0, VP1, VP2, FC0, FC1, EC0, and EC1 represent different estimator parameters. The true values are provided in the first row, and the second row presents mean estimates derived from 100 Monte Carlo simulations. Values in parentheses indicate standardized mean squared errors.\\
  \item[3] "Time" denotes the computational duration in seconds, where lower values suggest more efficient performance. For FD without $x_t$ dependence, computation times are significantly higher for moderate state space sizes (e.g., $X=2500$). Conversely, FD and FD2 estimators show much lower time costs.\\
  \item[4] $\tilde{\mathbf F}^{(2)}_1$ and $\tilde{\mathbf F}^{(2)}_2$ quantify the bias correction in the estimators, with values nearing zero signifying greater precision. In scenarios where $\gamma_a = 0.5$, these metrics indicate a noticeable bias in estimator performance, particularly for the FC0 coefficient.\\
  \item[5] The $\gamma_a$ parameter values of 0 and 0.5 delineate finite-dependence: 0 implies a 1-period finite dependence, whereas 0.5 suggests a departure from this pattern. For 1-period finite dependence models, all estimators are consistent. For non-1-period finite dependence models, only FD2 is consistent, with the FD biased.\\
  \item[6] The variable X represents the sample size used in Monte Carlo simulations, with larger values expected to produce more precise estimators.
\end{tablenotes}
\end{threeparttable}
\end{table}

\section{Extension to dynamic game}\label{sec:6_game}
In the context of game settings, individual players can exhibit finite dependence. This means that for each player, the state variables' transition matrices are influenced by the player's arbitrary weighting of their potential future decisions (as long as the weights sum up to one within a given period), while the other players adhere to their equilibrium strategies. Determining finite dependence in games cannot solely rely on the transition primitives, as in individual optimization. Whether finite dependence exists can also depend on which equilibrium is played. This is not a paradoxical outcome, as different equilibria in the same game can reveal varying information about the primitives, thus naturally requiring different estimation approaches.

In the games setting, assume there are $i$ players making choices in period $t = 1, \ldots, T.$ The systematic part of payoffs to the $i$-th player not only depends on his own choice in period $t$, $d_{it}$, the state variable $x_t$, but also the choice of the other players, denoted using $d_{t}^{(-i)} = (d_{t}^{(1)}, \ldots, d_{t}^{(i-1)}, d_{t}^{(i+1)}, \ldots, d_{t}^{(N)})$.
The flow utility of player $i$ is denoted by $U_{t}^{(i)}(x_t, d, d_{t}^{(-i)}) + \epsilon_{idt},$ where $\epsilon_{dt}^{(i)}$ denote the identically and independently distributed random variable that is private information to player $i$. 
Although the players all face the same observed state variables, these state variables typically affect players differently. 
For example, adding to the $i$-th player's capital may increase his payoffs and reduce the payoffs to the others. 
For this reason, the payoff function is indexed by $it.$

The players make simultaneous decisions in each period. We let $p_t(d_{t}^{(-i)}| x_t)$ denote the joint conditional choice probability that the players aside from $i$ collectively choose $d_{t}^{(-i)}$ at time $t$ given the state variable $x_t$. 
Since $\epsilon^{(i)}_t$ is independently distributed across all the players, $p_t(d_{t}^{(-i)}| x_t)$ has a product representation: 
\begin{equation}
p_t(d_{t}^{(-i)}| x_t) = \Pi_{\tilde i \neq i} \Big( p_t(d_{t}^{(\tilde i)}| x_t) \Big), \label{eq:game_belief}    
\end{equation}

We assume each player acts like a Bayesian when forming his beliefs about the choices of the other players and that a Markov-perfect equilibrium is played. Hence, the players' beliefs match the probabilities in equation \eqref{eq:game_belief}.
Taking the expectation of  $U_j^t(n)(x_t, d_t^{(-n)})$ over $d_t^{(-n)}$,  we define the systematic component of the current utility of player $i$ as a function of the state variables as
\begin{equation}
u^t(n)(x_t) = \mathbb{E}_{d_t^{(-n)} \in D^{N-1}} \left[P_t(d_t^{(-n)}\mid x_t) U_j^t(n)\mid x_t, d_t^{(-n)}\right]    
\end{equation}

The values of the state variables at period t+1 are determined by the period t choices by all the players as well as the values of the period t state variables. We consider a model in which the state variables can be partitioned into those that are affected by only one of the players and those that are exogenous. For example, to explain the number and size of firms in an industry, the state variables for the model might be indicators of whether each potential firm is active or not and a scalar to measure firm capital or capacity; each firm controls its own state variables, through their entry and exit choices, as well as their investment decisions.
The partition can be expressed as $x_t \equiv (z_t, \omega^{(1)}_t, ..., \omega^{(N)}_t)$, where $z_t$ denotes the states that are exogenously determined by transition probability $f_{z,t}(z_{t+1}|z_{t})$, and $\omega^{(i)}_t \in X^{(i)} \equiv {1, ..., X^{(i)}}$ is the component of the state-controlled or influenced by player $i$. Let $f^{(i)}_{\omega, t}(\omega^{(i)}_{t+1}|\omega^{(i)}_{t}, d)$ denote the probability that $\omega^{(i)}_{t+1}$ occurs at time $t+1$ when player $i$ chooses $d$ at time $t$ given $\omega^{(i)}_t$. Many models in industrial organizations exploit this specialized structure because it provides a flexible way for players to interact while keeping the model simple enough to be empirically tractable. 

Denote the state variables associated with all the players aside from $i$ as 
\[\begin{split}
\omega^{(-i)}_t \equiv \omega^{(1)}_t \times \ldots \times \omega^{(i-1)}_t \times \omega^{(i+1)}_t \times \ldots \times \omega^{(N)}_t  \\ \in  \mathcal{W}^{(- n)} \equiv \mathcal{W}^{(1)} \times \ldots \times \mathcal{W}^{(i-1)} \times \mathcal{W}^{(i+1)} \times \ldots \times \mathcal{W}^{(N)}.
\end{split}\] 
Under this specification, the reduced form transition generated by their equilibrium choice probabilities is defined as:
\[
f^{(- i)}_{\omega,t} \big( \omega^{(- i)}_{t+1} | \underbrace{ z_{t},\omega^{(i)}_{t}, \omega^{(- i)}_{t}}_{x_t} \big ) \equiv \prod_{\tilde i \neq i} \prod_{k=0}^D p_t(d_{t}^{(\tilde i)}| x_t) f^{(\tilde i)}_{\omega, t} \left( \omega^{(\tilde i)}_{t+1}|\omega^{(\tilde i)}_{t}, k \right)
\]

In the context of our game-theoretic model, we can delineate the state transition dynamics for each player with greater clarity. For player \(i\), the transition matrix is denoted by \(f^{(i)}( x_{t+1} | x_{t}, d)\). This matrix encapsulates the probability that player \(i\) will transition from their current state \(x_{t}\) to a subsequent state \(x_{t+1}\), contingent upon the decision \(d\). The composition of this transition matrix is tripartite, consisting of the following elements:

\begin{itemize}
  \item The opponent's strategy, represented by \(f^{(-i)}_{\omega,t} \big( \omega^{(- i)}_{t+1} | z_{t}, \omega^{(i)}_{t}, \omega^{(-i)}_{t} \big)\), which describes the likelihood of the opponent transitioning to a state \(\omega^{(- i)}_{t+1}\) based on the current joint state of the game.
  \item The stochastic process \(f_{z,t}(z_{t+1}|z_t)\) dictating the evolution of the game's environment from state \(z_t\) to state \(z_{t+1}\).
  \item The player's own strategy dynamics, given by \(f_{\omega,t}^{(i)}( \omega^{(i)}_{t+1} | x_t, d)\), which outlines how player \(i\)'s own state is expected to progress.
\end{itemize}

Upon determining these components, we aggregate the elements to construct the matrix \(\mathbf{F}_d^{(i)}\) for each possible decision \(d\), where \(d\) ranges from 0 to \(D\). Subsequently, for each player, we apply the methodology delineated in the preceding section to compute the weight matrix.

\begin{appendix}
\section{Proof}
  
\begin{proof}[Proof of Proposition \ref{prop:HM}]
    See proof of \citet[][Proposition 1]{Hotz1993}.
\end{proof}

\begin{proof}[Proof of Proposition \ref{prop:finite_dependence_two_period}]
  We can write the product $\mathbf{F}_0  \mathcal{P}_{\mathbf{\tilde F}}$ as:
  \[ 
  \mathbf{F}_0  \mathcal{P}_{\mathbf{\tilde F}} = \mathbf{V}_{\mathbf{\tilde F}} \mathbf{S}_0 \mathbf{V}_{\mathbf{\tilde F}}^\top  \mathbf{V}_{\mathbf{\tilde F}} \Sigma_{\text{Null}(\tilde F)} \mathbf{V}_{\mathbf{\tilde F}}^\top =  \mathbf{V}_{\mathbf{\tilde F}} \begin{bmatrix}
       \mathbf{0} &  \mathbf{S}_{0,01} \\
       \mathbf{0} &  \mathbf{S}_{0,11} \\
  \end{bmatrix} \mathbf{V}_{\mathbf{\tilde F}}^\top.
  \]
  
  For $\rho=2$, let's define $\mathbf{E} = \mathbf{\tilde F} (\mathbf{F}_0  \mathcal{P}_{\mathbf{\tilde F}})^{\rho-1} = \mathbf{U}_{\mathbf{\tilde F}} \mathbf{S}_{\mathbf{E}} \mathbf{V}_{\mathbf{\tilde F}}^\top$, with $\mathbf{S}_{\mathbf{E}} = \begin{bmatrix}
       \mathbf{0} &  \mathbf{S}_{\tilde F,00} \mathbf{S}_{0,01} \\
       \mathbf{0} &  \mathbf{0}  \\
  \end{bmatrix}$. 
  
  Since $\mathbf{S}_{\tilde F,00}$ has full rank, we have $\mathbf{I} - \mathbf{E}^+ \mathbf{E} =  \mathbf{V}_{\mathbf{\tilde F}} ( \mathbf{I} -  \begin{bmatrix}
       \mathbf{0} &  \mathbf{0}  \\
       \mathbf{0} &   \mathbf{S}_{0,01}^+ \mathbf{S}_{0,01}  \\
  \end{bmatrix} )   \mathbf{V}_{\mathbf{\tilde F}}^\top$.
  
  Now, we can express $\Delta \mathbf{F}^{(2)} = \mathbf{\tilde F} (\mathbf{F}_0 \mathbf{F}_0  \mathcal{P}_{\mathbf{\tilde F}}) \Big(\mathbf{I} - \mathbf{E}^+ \mathbf{E} \Big)$ as follows:
      \[\begin{split}
       \Delta \mathbf{F}^{(2)} & = \Big( \mathbf{U}_{\tilde F} \mathbf{S}_{\tilde F}  \mathbf{S}_{0} \mathbf{S}_{0} \Sigma_{\text{Null}(\tilde F)} \mathbf{V}_{\mathbf{\tilde F}}^\top \Big)  \Big( \mathbf{V}_{\mathbf{\tilde F}} ( \mathbf{I} -  \begin{bmatrix}
       \mathbf{0} &  \mathbf{0}  \\
       \mathbf{0} &   \mathbf{S}_{0,01}^+ \mathbf{S}_{0,01}  \\
      \end{bmatrix} )   \mathbf{V}_{\mathbf{\tilde F}}^\top \Big)\\
       & = \mathbf{U}_{\tilde F} \begin{bmatrix}
       \mathbf{S}_{\tilde F,00} \mathbf{S}_{0,00} &  \mathbf{S}_{\tilde F,00} \mathbf{S}_{0,01} \\
       \mathbf{0} &  \mathbf{0}  \\
      \end{bmatrix} \begin{bmatrix}
      \mathbf{0} & \mathbf{S}_{0,01} \\
      \mathbf{0} & \mathbf{S}_{0,11} \\
      \end{bmatrix} \begin{bmatrix}
       \mathbf{I} &  \mathbf{0}  \\
       \mathbf{0} &   \mathbf{I} - \mathbf{S}_{0,01}^+ \mathbf{S}_{0,01}  \\
      \end{bmatrix} \mathbf{V}_{\mathbf{\tilde F}}^\top
       \\
       &= \mathbf{U}_{\tilde F}  \begin{bmatrix}
      \mathbf{0} &  \mathbf{S}_{\tilde F,00} \Big( \mathbf{S}_{0,01} \mathbf{S}_{0,11} (\mathbf{I} - \mathbf{S}_{0,01}^+ \mathbf{S}_{0,01} ) \Big) \\
      \mathbf{0} &  \mathbf{0}  \\
      \end{bmatrix} \mathbf{V}_{\mathbf{\tilde F}}^\top. 
      \end{split} \]
     Since $\mathbf{S}_{\tilde F,00}$ has full rank, $\Delta \mathbf{F}^{(2)} = 0$ if $\mathbf{S}_{0,01} \mathbf{S}_{0,11} (\mathbf{I} - \mathbf{S}_{0,01}^+ \mathbf{S}_{0,01}) = 0$.
     \end{proof}

\end{appendix}
\bibliographystyle{ecta-fullname} 
\bibliography{reference2}  

\end{document}


\title{{\Large Appendix for "Estimate Non-Finite-Dependent Dynamic Discrete Choice Model with Unobserved Heterogeneity"}}
\author{Yu (Jasmine) Hao \footnote{\href{mailto: haoyu@hku.hk} {haoyu@hku.hk} HKU Business School, Hong Kong University.} \\ \emph{HKU Business School}
  \and Hiro Kasahara\footnote{ \href{mailto: jale@uandes.cl}{hkasahar@mail.ubc.ca} Vancouver School of Economics.}
   \\ \emph{University of British Columbia.}
  }
\date{\today}
\maketitle

\section{Additional Simulation Results}
\subsection{Two-step Estimator}
\begin{table}[H]
\centering
\scriptsize
\caption{Mean Estimates of Two-step estimator, $X=6250$, $N = 200, T = 50$}
\begin{threeparttable}
\begin{tabular}{lccccccc cc}
\toprule
 & $\theta^\mathrm{VP0}_m$  & $\theta^\mathrm{VP1}_m$ & $\theta^\mathrm{VP2}_m$& $\theta^\mathrm{FC0}_m$& $\theta^\mathrm{FC1}_m$& $\theta^\mathrm{EC0}_m$  & $\theta^\mathrm{EC1}_m$ & Time  &  $\rho$   \\
 \midrule
 \multicolumn{10}{c}{$\gamma_a = 0$(finite dependent model)}\\
 \midrule
$\bs{\theta}^{\mathrm{HM}}$                      & 0.499   & 1.001   & -0.999  & 0.500   & 0.998   & 1.003   & 0.994   & 15.740 & 0.000e+0  \\
& (0.027) & (0.038) & (0.032) & (0.047) & (0.053) & (0.055) & (0.084) & 15.350 &           \\
$\bs{\theta}^{\mathrm{AFD}}_{1,0}$               & 0.499   & 1.001   & -0.999  & 0.501   & 0.998   & 1.004   & 0.991   & 1.679  & 6.317e-6  \\
& (0.030) & (0.039) & (0.033) & (0.058) & (0.056) & (0.057) & (0.096) & 1.240  &           \\
$\bs{\theta}^{\mathrm{AFD}}_{1,q^{\mathrm{BE}}}$ & 0.499   & 1.001   & -0.999  & 0.501   & 0.998   & 1.004   & 0.991   & 32.484 & 6.317e-6  \\
& (0.030) & (0.039) & (0.033) & (0.058) & (0.056) & (0.057) & (0.096) & 1.240  &           \\
$\bs{\theta}^{\mathrm{AFD}}_{2,0}$               & 0.499   & 1.001   & -0.999  & 0.501   & 0.998   & 1.004   & 0.991   & 3.030  & 1.037e-11 \\
& (0.030) & (0.039) & (0.033) & (0.058) & (0.056) & (0.057) & (0.096) & 2.470  &           \\
$\bs{\theta}^{\mathrm{AFD}}_{3,0}$               & 0.499   & 1.001   & -0.999  & 0.501   & 0.998   & 1.004   & 0.991   & 4.374  & 1.720e-17 \\
& (0.030) & (0.039) & (0.033) & (0.058) & (0.056) & (0.057) & (0.096) & 3.679  &           \\
$\bs{\theta}^{\mathrm{AFD}}_{4,0}$               & 0.499   & 1.001   & -0.999  & 0.501   & 0.998   & 1.004   & 0.991   & 5.725  & 6.938e-19 \\
& (0.030) & (0.039) & (0.033) & (0.058) & (0.056) & (0.057) & (0.096) & 4.905  &           \\
$\bs{\theta}^{\mathrm{AFD}}_{5,0}$               & 0.499   & 1.001   & -0.999  & 0.501   & 0.998   & 1.004   & 0.991   & 7.062  & 6.848e-19 \\
& (0.030) & (0.039) & (0.033) & (0.058) & (0.056) & (0.057) & (0.096) & 6.109  &          \\     
\midrule
\multicolumn{10}{c}{$\gamma_a = 1$(non-finite dependent model)}\\
 \midrule
$\bs{\theta}^{\mathrm{HM}}$                      & 0.509   & 0.998   & -0.999  & 0.523   & 0.996   & 0.987   & 1.008   & 15.918 & 0.000 \\
& (0.041) & (0.030) & (0.030) & (0.103) & (0.055) & (0.075) & (0.095) & 15.647 &       \\
$\bs{\theta}^{\mathrm{AFD}}_{1,0}$               & 0.416   & 0.999   & -0.999  & 0.235   & 0.998   & 0.988   & 1.008   & 1.508  & 1.024 \\
& (0.094) & (0.031) & (0.032) & (0.286) & (0.059) & (0.077) & (0.106) & 1.065  &       \\
$\bs{\theta}^{\mathrm{AFD}}_{1,q^{\mathrm{BE}}}$ & 0.456   & 0.998   & -0.999  & 0.347   & 0.997   & 0.987   & 1.008   & 1.633  & 1.024 \\
& (0.063) & (0.031) & (0.032) & (0.191) & (0.059) & (0.077) & (0.106) & 1.065  &       \\
$\bs{\theta}^{\mathrm{AFD}}_{2,0}$               & 0.478   & 0.997   & -0.999  & 0.412   & 0.997   & 0.987   & 1.008   & 2.716  & 0.302 \\
& (0.053) & (0.031) & (0.032) & (0.153) & (0.059) & (0.077) & (0.106) & 2.149  &       \\
$\bs{\theta}^{\mathrm{AFD}}_{3,0}$               & 0.491   & 0.997   & -0.999  & 0.465   & 0.997   & 0.987   & 1.008   & 3.931  & 0.075 \\
& (0.050) & (0.031) & (0.032) & (0.135) & (0.059) & (0.077) & (0.106) & 3.232  &       \\
$\bs{\theta}^{\mathrm{AFD}}_{4,0}$               & 0.494   & 0.997   & -0.999  & 0.476   & 0.997   & 0.987   & 1.008   & 5.149  & 0.019 \\
& (0.050) & (0.031) & (0.032) & (0.134) & (0.059) & (0.077) & (0.106) & 4.313  &       \\
$\bs{\theta}^{\mathrm{AFD}}_{5,0}$               & 0.495   & 0.997   & -0.999  & 0.478   & 0.997   & 0.987   & 1.008   & 6.354  & 0.005 \\
& (0.050) & (0.031) & (0.032) & (0.134) & (0.059) & (0.077) & (0.106) & 5.384  &       \\
 \bottomrule
\end{tabular}
\begin{tablenotes}
    \item[1] The data is generated with $\theta = (\theta_0^{VP}=0.5, \theta_1^{VP}=1, \theta_2^{VP}=-1,\theta_0^{FC}=0.5,\theta_1^{FC}=1.0,\theta_0^{EC}=1.0,\theta_1^{EC})$.
    \item[2] The first row reports the mean of estimates across 100 Monte Carlo simulations, and the second row reports the standardized mean squared error of the estimates.
    \item[3] The second row of the time column with the HM estimator reports the time used for matrix inversion. The second rows of the AFD estimators reports the total time used to solve the optimal weight(weights). 
\end{tablenotes}
\end{threeparttable}\label{tab:two_step_large}
\end{table}

\subsection{Sequential Estimator for Two-type Finite Mixture Model}

\begin{table}[H]
\centering
\scriptsize
\caption{Mean Estimator and Squared Bias, $X=15552$, with $N = 200, T = 50$}
\begin{threeparttable}
\begin{tabular}{ll cccccccc}
\toprule
& & $\theta^\mathrm{VP0}_m$  & $\theta^\mathrm{VP1}_m$ & $\theta^\mathrm{VP2}_m$& $\theta^\mathrm{FC0}_m$& $\theta^\mathrm{FC1}_m$& $\theta^\mathrm{EC0}_m$  & $\theta^\mathrm{EC1}_m$ & $\alpha_m$  \\ \midrule
$\bs{\theta}^{\mathrm{NPL}}$    & Type 1 & 0.497   & 1.002   & -0.995  & 0.496   & 0.992   & 0.992   & 0.992   & 0.403  \\
 &   & (0.061) & (0.045) & (0.046) & (0.144) & (0.126) & (0.081) & (0.149) &   \\
 & Type 2 & 1.547   & 1.035   & -1.033  & 0.588   & 1.001   & 0.967   & 1.057   & 0.597  \\
 &   & (0.170) & (0.127) & (0.095) & (0.387) & (0.220) & (0.303) & (0.513) &   \\
 & \multicolumn{9}{l}{Time: 4133.250, Iter: 9.150}   \\
 \midrule
$\bs{\theta}^{\mathrm{AFD-SEQ}}_{1}$ & Type 1 & 0.454   & 1.002   & -0.994  & 0.378   & 0.993   & 0.992   & 0.993   & 0.403  \\
 &   & (0.076) & (0.045) & (0.045) & (0.188) & (0.126) & (0.080) & (0.152) &   \\
 & Type 2 & 1.444   & 1.033   & -1.032  & 0.285   & 1.002   & 0.960   & 1.064   & 0.597  \\
 &   & (0.171) & (0.126) & (0.093) & (0.419) & (0.220) & (0.307) & (0.518) &   \\
 & \multicolumn{9}{l}{Time: 651.808, Iter: 17.050, $\rho$: 1.375}   \\
 \midrule
$\bs{\theta}^{\mathrm{AFD-SEQ}}_{2}$ & Type 1 & 0.451   & 1.002   & -0.994  & 0.370   & 0.993   & 0.992   & 0.993   & 0.403  \\
 &   & (0.078) & (0.045) & (0.045) & (0.193) & (0.126) & (0.080) & (0.152) &   \\
 & Type 2 & 1.442   & 1.033   & -1.032  & 0.278   & 1.002   & 0.960   & 1.064   & 0.597  \\
 &   & (0.171) & (0.126) & (0.093) & (0.422) & (0.220) & (0.307) & (0.518) &   \\
 & \multicolumn{9}{l}{Time: 1633.743, Iter: 9.000, $\rho$: 0.369}   \\
 \midrule
$\bs{\theta}^{\mathrm{AFD-SEQ}}_{3}$ & Type 1 & 0.458   & 1.002   & -0.994  & 0.400   & 0.993   & 0.991   & 0.993   & 0.403  \\
 & MBE 1  & (0.074) & (0.045) & (0.046) & (0.175) & (0.126) & (0.080) & (0.152) &   \\
 & Type 2 & 1.459   & 1.033   & -1.032  & 0.350   & 1.002   & 0.961   & 1.062   & 0.597  \\
 & MBE 2  & (0.167) & (0.126) & (0.093) & (0.390) & (0.220) & (0.306) & (0.516) &   \\
 & \multicolumn{9}{l}{Time: 1282.364, Iter: 8.300, $\rho$: 0.084}  \\
\bottomrule
\end{tabular}
\begin{tablenotes}
    \item[1] $nM=200, nT=50$, based on $20$ Monte Carlo simulations. 
    \item[2] The true parameters are $\theta_1 = (\theta_0^{VP}=0.5, \theta_1^{VP}=1, \theta_2^{VP}=-1,\theta_0^{FC}=0.5,\theta_1^{FC}=1.0,\theta_0^{EC}=1.0,\theta_1^{EC})=1.0$, $\theta_2 = (\theta_0^{VP}=1.5, \theta_1^{VP}=1, \theta_2^{VP}=-1,\theta_0^{FC}=0.5,\theta_1^{FC}=1.0,\theta_0^{EC}=1.0,\theta_1^{EC})=1.0$.
\end{tablenotes}
\end{threeparttable}\label{tab:sequential_large_2}
\end{table}

\section{Characterization of Weight Solving}
\subsection{Two periods "almost finite dependence"}\label{sec:weight_2_periods}
It is possible for us to maximize two-period weight simultaneously so that the objective function becomes

\[ \min_{\mathbf{w}_{1}, \mathbf{w}_{2}} \lVert \mathbf{\tilde F} \mathbf{F}(\mathbf{w}_{1}) \mathbf{F}(\mathbf{w}_{2}) \rVert \]
Note that $\sum_{d^{\dagger} \in \mathcal{D} } \mathbf{w}_{1}(d^{\dagger}, i^{\dagger}) f ( i^{\dagger\dagger} | d^{\dagger}, i^{\dagger} ) = \sum_{d^{\dagger} \in \mathcal{D} /\{ 0 \}} \mathbf{w}_{1}(d^{\dagger}, i^{\dagger}) \tilde f ( i^{\dagger\dagger} | d^{\dagger}, i^{\dagger} ) + f(i^{\dagger\dagger} | 0, i^{\dagger}  ) $.
The $(d,i,j)$-th element of the objective function, which is the $(ij)$-th element of $ \mathbf{\tilde F}(d) \mathbf{F}(\mathbf{w}_{1}) \mathbf{F}(\mathbf{w}_{2})$ is 
\begin{equation}
    \begin{split}
  & [\mathbf{\tilde F}(d) \mathbf{F}(\mathbf{w}_{1}) \mathbf{F}(\mathbf{w}_{2})]_{ij}  \\ & = \sum_{\substack{i^{\dagger} \in \mathcal{X}, \\ i^{\dagger\dagger} \in \mathcal{X}  } } \tilde f(i^{\dagger}|d,i) \underbrace{ \left( \sum_{d^{\dagger} \in \mathcal{D}/\{ 0 \}} \mathrm{w}^{1}(d^{\dagger}, i^{\dagger}) \tilde f ( i^{\dagger\dagger} | d^{\dagger}, i^{\dagger}) + f(i^{\dagger\dagger} | 0, i^{\dagger}  )  \right)}_{ f^{\mathbf{w}_1}( i^{\dagger \dagger} | i^{\dagger} ) }  \\ & \times \underbrace{\left( \sum_{d^{\dagger\dagger }  \in \mathcal{D}/\{ 0 \} }\mathrm{w}^{2}(d^{\dagger\dagger}, i^{\dagger\dagger}) \tilde f ( j | d^{\dagger\dagger}, i^{\dagger\dagger} ) +  f(j | 0, i^{\dagger\dagger}  ) \right)}_{f^{\mathbf{w}_2}( j | i^{\dagger\dagger} ) } \\
  & = (\mathbf{w}_1^{+})\t \mathbf{H}(d,i,j) \mathbf{w}_2^{+} + (\mathbf{w}_1^{+})\t \mathbf{h^1}(d,i,j) + (\mathbf{w}_2^{+})\t \mathbf{h^2}(d,i,j) + \mathbf{h^0}(d,i,j).
    \end{split}
\end{equation}
where $\mathbf{w}_1^{+} = [\mathrm{w}^{1}(1,1)\ldots, \mathrm{w}^{1}(1,X),\ldots, \mathrm{w}^{1}(D,1)\ldots, \mathrm{w}^{1}(D,X)]\t$ and  \\
$\mathbf{w}_2^{+} = [\mathrm{w}^{2}(1,1)\ldots, \mathrm{w}^{2}(1,X),\ldots, \mathrm{w}^{2}(D,1)\ldots, \mathrm{w}^{2}(D,X)]\t$ are $ \big(DX\big) \times 1$ vectors such that for each state $x$, we leave $w(0,x)$ out.
\begin{equation}
    \begin{split}
  \underbrace{\mathbf{H}(d,i,j)}_{ (XD) \times (XD) \text{ matrix } } & = \Big[\tilde f(i^\dagger | d,i) \tilde f (i^{\dagger \dagger }| d^{\dagger},i^{\dagger}) \tilde f(j |d^{\dagger\dagger},i^{\dagger\dagger} ) \Big]_{\big(d^{\dagger},i^{\dagger}\big), \big(d^{\dagger\dagger},  i^{\dagger\dagger}\big) }, \\ 
  \underbrace{\mathbf{h^1}}_{ (XD) \times 1 \text{ row vector } } & = \Big[ f(i^\dagger | d,i) \sum_{i^{\dagger \dagger } \in \mathcal{X}} \tilde f (i^{\dagger \dagger }| d^{\dagger},i^{\dagger})  f(j |0,i^{\dagger\dagger} )  \Big]_{\big(d^{\dagger},i^{\dagger}) \big) } , \\ 
  \underbrace{\mathbf{h^2}}_{ (XD) \times 1 \text{ row vector } } & = \Big[ \sum_{i^{\dagger } \in \mathcal{X}} f(i^\dagger | d,i)    f (i^{\dagger \dagger }| 0,i^{\dagger}) \tilde f(j |d^{\dagger \dagger},i^{\dagger\dagger} )  \Big]_{\big(d^{\dagger\dagger},i^{\dagger\dagger}) \big) }, \\
  \mathrm{h}^0 & = \sum_{\substack{i^{\dagger} \in \mathcal{X},  \\ i^{\dagger\dagger} \in \mathcal{X}  } } f(i^\dagger | d,i)  f (i^{\dagger \dagger }| 0,i^{\dagger}) \tilde f(j | 0,i^{\dagger\dagger} ).
    \end{split}
\end{equation}
Note that a 1-period finite dependence model is equivalent to the condition that that there exists $\mathbf{w}_1^{+}, \mathbf{w}_2^{+}$ such that:
\[ \sum_{\substack{d\in\mathcal{D} /\{0\},\\ i\in\mathcal{X}, \\ j \in\mathcal{X}}} \left( (\mathbf{w}_1^{+})\t \mathbf{H}(d,i,j) \mathbf{w}_2^{+} + (\mathbf{w}_1^{+})\t \mathbf{h^1}(d,i,j) + (\mathbf{w}_2^{+})\t \mathbf{h^2}(d,i,j) + \mathbf{h^0}(d,i,j) \right)^2 = 0.\]
With the above expression, the objective of the minimization is 
\begin{equation*}
    \begin{split}
  \min_{ \mathbf{w}_1^{+}, \mathbf{w}_2^{+} } \lVert \tilde{\mathbf F} \mathbf{F}(\mathbf w^1)  \mathbf{F}(\mathbf w^2) \rVert &  = \min_{ \mathbf{w}_1^{+}, \mathbf{w}_2^{+} }  \sum_{\substack{d\in\mathcal{D} /\{0\},\\ i\in\mathcal{X}, \\ j \in\mathcal{X}}} \Big( (\mathbf{w}_1^{+})\t \mathbf{H}(d,i,j) \mathbf{w}_2^{+} \\ & + (\mathbf{w}_1^{+})\t \mathbf{h^1}(d,i,j) + (\mathbf{w}_2^{+})\t \mathbf{h^2}(d,i,j) + \mathbf{h^0}(d,i,j) \Big)^2 .
    \end{split}
\end{equation*} 
\bibliographystyle{apalike}
\bibliography{reference2}